\documentclass[10pt]{IEEEtran}
\onecolumn
\pdfoutput=1

\usepackage[colorlinks=true,pdfstartview=FitV,linkcolor=black,citecolor=black,urlcolor=blue,plainpages=false]{hyperref}

\usepackage[numbers,sort&compress]{natbib}

\usepackage[T1]{fontenc}
\usepackage[latin1]{inputenc}

\usepackage{amsmath}
\usepackage{amsthm}
\usepackage{amsfonts}
\usepackage{dsfont}
\usepackage{mathrsfs}
\usepackage{fullpage}
\usepackage{enumerate}

\usepackage{tikz}

\usepackage[]{authblk}
\usepackage{times}
\usepackage{epsf}
\usepackage{psfrag}
\usepackage{epsfig}

\usepackage{amssymb}
\usepackage{amstext}
\usepackage{latexsym}
\usepackage{color}
\usepackage{ifthen}
\usepackage{multirow}

\usepackage{subfigure}

\renewcommand{\vec}[1]{{\bf #1}}

\renewcommand{\L}{{\mathcal L}}

\newcommand{\Z}{{\mathds Z}}

\newcommand{\ep}{\epsilon}
\newcommand{\vep}{\varepsilon}
\newcommand{\C}{{\mathcal C}}
\newcommand{\B}{{\mathcal B}}

\newcommand{\T}{{\mathcal T}}
\newcommand{\I}{{\mathcal I}}
\newcommand{\N}{{\mathcal N}}
\newcommand{\IN}{{\mathbb N}}
\newcommand{\R}{{\mathbb R}}

\newcommand{\Ps}{{\mathcal P}}

\newcommand{\one}{{\mathds 1}}

\newcommand{\SNR}{{\rm SNR}}

\newcommand{\vectil}[1]{{\bf \tilde #1}}

\newcommand{\defi}{\triangleq}

\newcommand{\avg}{{\rm avg}}

\newcommand{\goesto}{\rightarrow}

\newcommand{\st}{:}

\newcounter{constcount}
\newcounter{numcount}
\newenvironment{lemmarep}[1]{\noindent {\bf Lemma #1.}\begin{it}}{\end{it}}

\newcommand{\eqnum}{\stackrel{(\roman{numcount})}{=}\stepcounter{numcount}}

\newcommand{\leqnum}{\stackrel{(\roman{numcount})}{\leq\;}\stepcounter{numcount}}

\newcommand{\geqnum}{\stackrel{(\roman{numcount})}{\geq\;}\stepcounter{numcount}}

\newcommand{\cnt}{$(\roman{numcount})$\;\stepcounter{numcount}}

\newcommand{\rescnt}{\setcounter{numcount}{1}}

\newcommand{\IDFT}{\text{IDFT}}
\newcommand{\DFT}{\text{DFT}}

\newcounter{thmcnt}
  \let\Oldsection\section
\renewcommand{\section}{\stepcounter{thmcnt}\Oldsection}

\renewcommand{\subset}{\subseteq}

\newtheorem{theorem}{Theorem}%[thmcnt]%[section]

\newtheorem{lemma}{Lemma}

\newenvironment{remark}{\noindent \emph{Remark:}\,}{\vspace{2mm}}

\newcounter{examplecounter}
\newenvironment{example}
{
\stepcounter{examplecounter} {\vspace{4mm} \noindent \bf Example \arabic{examplecounter}.} \begin{it} \rm
}
{
\end{it} \vspace{2mm}}

\newcounter{defncounter}
\newenvironment{definition}
{
\refstepcounter{defncounter} {\vspace{4mm} \noindent \bf Definition \arabic{defncounter}.} \begin{it} \rm
}
{
\end{it} \vspace{2mm}}

\newcommand{\aln}[1]{\begin{align*}#1\end{align*}}

\newcommand{\al}[1]{\begin{align}#1\end{align}}

\begin{document}

%\title{Is Gaussian Noise the Worst-Case Additive Noise in Networks?}
%\title{Worst-Case Additive Noise in Networks}
\title{Worst-Case Additive Noise in \\ Wireless Networks}

%\author{Ilan Shomorony and A. Salman Avestimehr\\ Cornell University, Ithaca, NY}

\author{Ilan Shomorony and A. Salman Avestimehr \thanks{The authors are with the School of Electrical and Computer Engineering, Cornell University, Ithaca, NY 14853 USA (e-mails: is256@cornell.edu, avestimehr@ece.cornell.edu). \\
\indent This work is in part supported by NSF grants CAREER-0953117, CCF-1144000, CCF-1161720, AFOSR Young Investigator Program award FA9550-11-1-0064, and NSF TRUST Center. \\
%\indent Manuscript received March 14, 2011; revised January 30, 2012. Date of current version August 14, 2012. Communicated by S. Jafar, Associate Editor for Communications. \\ 
\indent This paper was presented in part at the International Symposium on Information Theory 2012, Cambridge, USA \cite{wcnoiseisit}}}

%\affil{School of Electrical and Computer Engineering\\ Cornell University, Ithaca, NY}
%\affil{Cornell University, Ithaca, NY}

\maketitle

%\noindent {\today}
%
%\vspace{2mm}
%
%{\large \begin{center} \bf Worst-Case Noise in Networks \end{center}}
%
%\vspace{2mm}

%{\bf 
%\begin{itemize}
%\item define AWGN network
%\end{itemize}}

%\begin{abstract}
%It is an important open question whether Gaussian noise is the worst-case additive noise in wireless networks.
%In this work, we show that, under very mild assumptions from a practical point of view, Gaussian noise is indeed the worst-case additive noise in general wireless networks with additive noises that are independent from the transmit signals.
%In particular, we show that, given a coding scheme with finite reading precision for an AWGN network, one can build a coding scheme that achieves the same rates on an additive noise wireless network with the same topology, where the noise terms may have any distribution with same mean and variance as in the AWGN network.
%\end{abstract}

\begin{abstract}
A classical result in Information Theory states that the Gaussian noise is the \emph{worst-case additive noise} in point-to-point channels, meaning that, for a fixed noise variance, the Gaussian noise minimizes the capacity of an additive noise channel.
In this paper, we significantly generalize this result and show that the Gaussian noise is also the worst-case additive noise in wireless networks with additive noises that are independent from the transmit signals. 
%More specifically, we prove that, given a coding scheme for an AWGN network, one can build a coding scheme that achieves the same rates on an additive noise wireless network that has the same topology, where the noise terms may have any distribution with same mean and variance as in the AWGN network.
More specifically, we show that, if we fix the noise variance at each node, then the capacity region with Gaussian noises is a subset of the capacity region with any other set of noise distributions.
We prove this result by showing that a coding scheme that achieves a given set of rates on a network with Gaussian additive noises can be used to construct a coding scheme that achieves the same set of rates on a network that has the same topology and traffic demands, but with non-Gaussian additive noises.
\end{abstract}

\begin{section}{Introduction} \label{intro}

%In this work, we 

The modeling of background noise in point-to-point wireless channels as an additive Gaussian noise is well supported from both theoretical and practical viewpoints. 
In practice, we have witnessed that current wireless systems that were designed based on the assumption of additive Gaussian noise perform quite well. 
This is intuitively explained by the fact that, from the Central Limit Theorem, the composite effect of many (almost) independent noise sources (e.g., thermal noise, shot noise, etc.) should approach a Gaussian distribution. 
From a theoretical point of view, Gaussian noise has been proven to be the \emph{worst-case noise} for additive noise channels. 
This means that, given a variance constraint, the Gaussian noise \emph{minimizes the capacity} of a point-to-point additive noise channel.
This result follows mainly from the fact that  the Gaussian distribution maximizes the entropy subject to a variance constraint.
More precisely, from the Channel Coding Theorem \cite{CoverThomas}, the capacity of a channel $f(y|x)$ is given by
\al{ \label{capexp}
C = \max_{f(x) : E[X^2]\leq P} I(X;Y).
}
Thus, if we choose $X$ to be distributed as $\N(0,P)$, we have that
\aln{
C \geq h(X)-h(X|Y) = \frac12 \log \left(2\pi e P\right) - h(X|Y).
}
As shown in \cite{CoverThomas}, for an additive noise (AN) channel $Y = X+Z$, where $E[Z] = 0$ and $E\left[ Z^2 \right] = \sigma^2$, we have $h(X|Y) \leq \frac12 \log \left(2 \pi e \frac{P\sigma^2}{P+\sigma^2}\right)$. 
We conclude that
\aln{
C_{\rm AN} \geq \frac12 \log \left(1+\frac{P}{\sigma^2}\right) = C_{\rm AWGN},
}
where $C_{\rm AWGN}$ is the capacity of the AWGN channel, which is achieved by a Gaussian input distribution.
Moreover, a more operational justification of the fact that Gaussian is the worst-case noise for additive noise channels was provided in  \cite{LapidothNearest}, where it was shown that random Gaussian codebooks and nearest-neighbor decoding achieve the capacity of the corresponding AWGN channel on a non-Gaussian AN channel. %(where the additive noise has the same mean and variance as in the Gaussian case).

Worst-case noise characterizations in settings other than a simple scalar additive noise channel are few in the literature.
One such example is \cite{DiggaviWorstCase}, where the authors consider vector channels with additive noise subject to the constraint that the noise covariance matrix lies in a convex set.
It is shown that, in this setting, the worst-case noise is vector Gaussian with a covariance matrix that depends on the transmit power constraints.
In \cite{ShamaiWorstCase}, a scalar additive noise channel with binary input is considered.
%In this work, worst-case noise is used to refer both to the noise distribution that minimizes the capacity and the noise distribution that %
The probability mass function of the (discrete) worst-case noise is characterized, and the worst-case capacity (i.e., the capacity under the worst-case noise) is found.
%Extensions of this worst-case noise result for other settings
Once we go beyond point-to-point channels, Gaussian noise is only known to be the worst-case additive noise in some special wireless networks, such as the Multiple Access Channel, the Degraded Broadcast Channel and MIMO channels. 
In all such cases the capacity has been fully characterized and is known to be achievable with Gaussian inputs. Therefore, similar arguments to the one above can be used to show that, in these cases, Gaussian noise is indeed the worst-case additive noise.
However, for more general wireless networks where the capacity is unknown, we lack the tools to make such an assertion.
The recent constant-gap capacity approximations for the Interference Channel \cite{ETW} and for single-source single-destination relay networks \cite{ADTFullPaper,NoisyNetworkCoding,SuhasLatticeRelay} can only be used to state that Gaussian noise is ``approximately'' the worst-case additive noise in these cases.
%In the case of relay networks, for example, this follows from the fact that the capacity is shown to be within a constant gap of the cut-set outer bound. 
%Since the cut-set bound 
%However,
%In these cases, as well as in most other additive noise wireless networks, it remains an open question whether Gaussian noise is the worst-case additive noise.
Nonetheless, in a leap of faith, most of the research concerning such systems and many other wireless networks views the AWGN channel model as the standard wireless link model.
In general, it remains unknown whether Gaussian noise is the worst-case additive noise in wireless networks.

In this work, we address this issue and show that the Gaussian noise is in fact the worst-case noise for arbitrary wireless networks with additive noises that are independent of the transmit signals.
We consider wireless networks with unrestricted topologies and general traffic demands.
We show that any coding scheme 
%with \emph{finite reading precision} 
that achieves a given set of rates on a network with Gaussian additive noises can be used to construct a coding scheme that achieves the same set of rates on a network that has the same topology and traffic demands, but with non-Gaussian additive noises.
It is also important to notice that our coding scheme construction only depends on the mean and variance of the noise distributions of our non-Gaussian network, and is oblivious to their precise statistics.
This means that our approach also results in a framework to design codes for networks with unknown noise distributions with an asymptotic performance guarantee.

We prove that the Gaussian noise is the worst-case noise in wireless networks based on two main results.
The first one is that, given a coding scheme with \emph{finite reading precision} for an AWGN network, one can build a coding scheme that achieves the same rates on a non-Gaussian wireless network.
A coding scheme is said to have finite reading precision if, for any node, its transmit signals only depend on its received signals read up to a finite number of digits after the decimal point.
%Notice that this precision can be chosen arbitrarily large, as long as it is finite, and is allowed to tend to infinity along a sequence of coding schemes.
%%It is clear, therefore, that such a restriction is very reasonable from a practical point of view.
%%Thus, our result implies that, when we restrict ourselves to ``practical schemes'', Gaussian noise is the worst-case noise in networks.
%Hence, this is a very mild restriction, and, in practice, almost all coding schemes satisfy it.
This result is proven in three main steps.
We start by applying a transformation at the transmit signals and received signals of all nodes in the network in order to create an ``approximately Gaussian'' effective network.
The technique resembles OFDM in that it uses the Discrete Fourier Transform in order to mix together multiple uses of the same channel.
This mixing causes the additive noise terms from distinct network uses to be averaged over time and, by making use of Lindeberg's Central Limit Theorem \cite{BillingsleyConvergence}, it can be shown that the resulting effective noise is approximately Gaussian in the distribution sense.
Thus, we create an approximately Gaussian network.
However, this mixing causes distinct noise realizations at the same receiver to be dependent of each other. 
The second step is an interleaving technique, which allows us to handle this dependence between distinct noise realizations.
%Secondly, we need to handle the dependence between noise realizations over time.
%This is done by combining a
The interleaving operation creates multiple blocks of network uses inside which the additive noises are i.i.d. and almost normally-distributed.
Inside each of these blocks we are able to apply the original coding scheme that we have for the AWGN network.
The third step involves evaluating the performance of our original coding scheme on this i.i.d.~almost normally-distributed blocks.
This can be done because we require the original coding scheme to have finite reading precision.
For such coding schemes, the sets of noise realizations that cause the coding scheme to make an error can be shown to be \emph{continuity sets}.
It follows from the portmanteau Theorem \cite{BillingsleyConvergence} that the coding scheme's performance on an almost-Gaussian network does not deviate much from its performance on an actual Gaussian network.

The second main result we need is that, for any wireless network, the capacity when we restrict ourselves to coding schemes with finite reading precision, and allow the precision to tend to infinity along the sequence of coding schemes, is the same as the unrestricted capacity.
To prove this we show that, for any coding scheme with infinite precision, there exists a quantization scheme of the received signals which does not increase the error probability of the coding scheme too much.
This is done by showing that a truncation of the bit expansion of the received signal followed by a random shift performs well; thus, there must exist a fixed shift for each node which guarantees the same performance.
This quantization operation makes the coding scheme have finite reading precision, and the result follows.

The paper is organized as follows. 
In Section \ref{setupsec}, we describe the network model and introduce the necessary terminology.
We start by focusing on wireless networks with $L$ unicast sessions, which makes the proofs simpler and easier to follow.
In Section \ref{mainsec}, we state our main result (Theorem \ref{mainthm}) and the two main theorems that are needed for it, in the context of $L$-unicast wireless networks.
Theorem \ref{finitelem1} states that coding schemes with finite reading precision can be used to construct coding schemes for non-Gaussian networks.
Theorem \ref{finitelem2} states that, for AWGN networks, coding schemes with infinite reading precision can be ``quantized'' yielding coding schemes with finite reading precision that perform almost as well.
We then state our main result for networks with general traffic demands (Theorem 4).
The proof of Theorem \ref{finitelem1} is presented in Section \ref{proofl}, divided into three subsections as follows.
%In the subsequent subsections we prove this result in three parts.
We first describe the OFDM-like scheme in subsection \ref{ofdmsec}.
Then, in Section \ref{noisesec}, we show that the additive noises obtained from the OFDM-like scheme in fact converge in distribution to Gaussian noises. 
In Section \ref{outercodesec}, we describe the interleaving technique and the outer code that are used to handle the dependence between the noises after the OFDM-like scheme, and we show how the requirement of finite reading precision can be used to show that our coding scheme designed for a Gaussian network can be applied to an almost-Gaussian network without much loss in performance.
The proof of Theorem \ref{finitelem2} is in Section \ref{finitesec}.
%Finally, in Section \ref{extensionsec}, we describe two extensions.
%First we show how our main result can be extended to the large class of ``truncatable'' coding schemes.
%Then we show that by making slight changes to the arguments in Sections \ref{ofdmsec}, \ref{noisesec} and \ref{outercodesec}, it is possible to extend our result to wireless networks with general traffic demands.
In Section \ref{extsec}, we describe how we can modify the arguments in the previous Sections in order to consider, instead of $L$-unicast wireless networks, wireless networks with general traffic demands, proving Theorem \ref{mainthm2}.
We conclude the paper in Section \ref{conclusion}.

\end{section}

\begin{section}{Problem Setup and Definitions} \label{setupsec}

In this work, we model wireless networks as follows. %with broadcast and (linear) interference.
%We will use the following definition.

\begin{definition}
An additive noise wireless network consists of a directed graph $G = (V,E)$, where $V$ is the vertex (or node) set and $E \subset V \times V$ is the edge set, and a real-valued channel gain $h_{u,v}$ associated with each edge $(u,v) \in E$. 
%Since we will often be referring to vertices by $v_i$, for $i \in \IN$, we will also use $h_{i,j}$ to represent the channel gain associated with edge $(v_i,v_j)$. 
%Communication in a wireless network is performed over a block of $n$ discrete time steps.
At time $t=0,1,2,...$, each node $u \in V$ transmits a real-valued signal $X_{u}[t]$. %, which must satisfy an average power constraint $\tfrac1n\sum_{t=1}^n X_u^2[t] \leq P$, $\forall \, u \in V$, for some fixed $P \geq 0$. %, where the expectation is taken with respect to any possible randomness involved. 
The signal received by node $v$ at time $t$ is given by
\al{ \label{channelmodel}
Y_v[t] = \sum_{u \in \I(v)} h_{u,v} X_u[t] + N_v[t],
}
where $\I(v) = \{ u \in V \st (u,v) \in E\}$, and the additive noise $N_v$ is assumed to be i.i.d.~over time and to satisfy $E[N_v]=0$ and $E\left[N_v^2\right]= \sigma_v^2 < \infty$.
We also assume that the noise terms are independent from all transmit signals and from all noise terms at distinct nodes. %, and that each $N_v$ has an absolutely continuous distribution.
If all the additive noises in the network are normal $\N(0,\sigma_v^2)$, then we say the network is an AWGN network.
\end{definition}

In order to define source-destination relationships in a wireless network, we introduce the following notion.

\begin{definition} \label{trafficdef}
For a wireless network with graph $G= (V,E)$, the \emph{traffic demand} is described by a function $\T : V \times \Ps(V) \to \{0,1\}$,
%vector $\T \in \{0,1\}^{V\times \Ps(V)}$, 
where $\Ps(V)$ is the power set of $V$.
For $s \in V$ and $D \subseteq V$, $\T(s,D) = 1$ if $s$ has a message that is required by all nodes in $D$ and no node outside of $D$, and $\T(s,D) = 0$ otherwise.
%The entry of $\T$ corresponding to a node $s \in V$ and a subset $D \subseteq V$ equals $1$ if $s$ has a message that is required by all nodes in  $D$, and $0$ otherwise.
\end{definition}

\begin{example}
An $L$-user multiple access channel is defined by a graph $G = (V,E)$ with node set $V = \{s_1,s_2,...,s_L,d\}$, edge set $E = \left\{ (s_1,d),...,(s_L,d) \right\}$, and traffic demands
\aln{
\T(v,U) = \left\{ 
\begin{array}{ll}
1 & \text{ if $v \in \{s_1,s_2,...,s_L\}$ and $U = \{d\}$} \\
0 & \text{ otherwise.}
\end{array}
\right.
 } 
\end{example}

\begin{example}
An $L$-user broadcast channel with degraded message sets is defined by a graph $G = (V,E)$ with node set $V = \{s,d_1,d_2,...,d_L\}$, edge set $E = \left\{ (s,d_1),...,(s,d_L) \right\}$, and traffic demands
\aln{
\T(v,U) = \left\{ 
\begin{array}{ll}
1 & \text{ if $v = s$ and $U = \{d_1,d_2,...,d_\ell\}$, for $\ell=1,2,...,L$} \\
0 & \text{ otherwise.}
\end{array}
\right.
 } 
\end{example}

Even though the results presented in this paper hold for wireless networks with any traffic demands, including the multiple access channel and the broadcast channel, we start by considering the following special class.

\begin{definition}
An $L$-unicast wireless network has $L$ source nodes $s_1,...,s_L \in V$ and $L$ destination nodes $d_1,...,d_L \in V$ all of which are distinct nodes, and traffic demands given by
\aln{
\T(v,U) = \left\{ 
\begin{array}{ll}
1 & \text{ if $(v,U)  = \left(s_\ell,\{d_\ell\}\right)$, for $\ell=1,2,...,L$} \\
0 & \text{ otherwise.}
\end{array}
\right.
}
\end{definition}
%
%focusing on \emph{$|L|$-unicast additive noise wireless networks}.
%The traffic demands in these networks are defined by a set $L$ of source-destination pairs, and we assume that all $|L|$ sources and $|L|$ destinations are distinct nodes.
%Therefore, classical examples of wireless networks such as the Multiple Access Channel (where multiple sources communicate with a single destination) and the Broadcast Channel (where a single source communicates with multiple destinations) do not fall in this category.

Presenting our results for $L$-unicast wireless networks first has the advantage of making some of the proofs simpler and easier to follow.
Later, in Section \ref{extsec}, we describe how the same results can be extended to wireless networks with an arbitrary traffic demand $\T$.

We point out that perfect (noiseless) feedback from a destination to a source is not allowed in our model.
However, in Section \ref{extsec}, we consider a generalization of Definition \ref{codedef} that allows the sources' transmit signals to depend on their previously received signals.
Thus, noisy feedback links may exist between a destination and its corresponding source, and by setting the noise variance at the source to be very small, nearly perfect feedback can be simulated.

%We consider a wireless relay network defined by a directed graph $G=(V,E)$, with $K$ source-destination pairs and an arbitrary topology.
%Each relay $v$ receives 
%\[ Y_v =  \sum_{u \in \I(v)} h_{u,v} X_u +N_v ,\]
%where the additive noise $N_v$ is i.i.d. and satisfies $E[N_v]=0$ and $E\left[N_v^2\right]=1$.
%First we consider the following definitions.

\begin{definition} \label{codedef}
A coding scheme $\C$ with block length $n \in \IN$ and rate tuple $\vec R = (R_1,...,R_{L}) \in \R^{L}$ for an $L$-unicast additive noise wireless network consists of:
\begin{enumerate}[1. ]
\item An encoding function $f_i : \{1,...,2^{n R_i}\} \to \R^{n}$ for each source $s_i$, $i=1,...,L$, where each codeword $f_i(w_i)$, $w_i \in \{1,...,2^{n R_i}\}$, satisfies an average power constraint of $P$.
\item Relaying functions $r_v^{(t)} : \R^{t-1} \to \R$, for $t=0,...,n-1$, for each node $v \in V$ that is not a source, satisfying the average power constraint
\aln{
\frac1n\sum_{t=0}^{n-1} \left[ r_v^{(t)}(y_0,...,y_{t-1})\right]^2 \leq  P, \quad \text{for all $(y_0,...,y_{n-1}) \in \R^{n}$.}
}
%where $Y_v^{t} = \left( Y_v[1],...,Y_v[t]\right)$, and
%where the expectation is taken assuming that each source chooses its message uniformly at random, and taking into account all noises in the network.
\item A decoding function $g_i : \R^n \to \{1,...,2^{n R_i}\}$ for each destination $d_i$, $i=1,...,L$.
\end{enumerate}
%We say that $n$ is the block length of the coding scheme $\C$, and $\vec R$ its rate tuple.
\end{definition}

\begin{definition}
The error probability of a coding scheme $\C$ (as defined in Definition \ref{codedef}), is given by
\aln{
P_{\rm error}(\C) = \Pr \left[ \bigcup_{i=1}^{L} \{ W_i \ne g_i(Y_{d_i}[0],...,Y_{d_i}[n-1]) \} \right],
}
where the message transmitted by source $s_i$, $W_i$, is assumed to be chosen uniformly at random from $\{1,...,2^{n R_i}\}$, for $i=1,...,L$.
\end{definition}

\begin{definition} \label{achievedef}
A rate tuple $\vec R$ is said to be achievable for an $L$-unicast wireless network if there exists a sequence of coding schemes $\C_n$ with rate tuple $\vec R$ and block length $n$, for which $P_{\rm error}(\C_n) \to 0$, as $n\to \infty$.
%the probability that at least one decoder makes an error tends to zero, i.e.,
%\aln{
%\Pr \left[ \bigcup_{i=1}^{|L|} \{ W_i \ne g_i(Y_{d_i}[1],...,Y_{d_i}[n]) \} \right] \to 0, \text{ as } n \to \infty,
%}
%where the message transmitted by source $s_i$, $W_i$, is assumed to be chosen uniformly at random from $\{1,...,2^{n R_i}\}$, for $i=1,...,|L|$.
The sequence of coding schemes $\C_n$, $n=1,2,...$, is then said to achieve rate tuple $\vec R$.
The capacity region of an $L$-unicast wireless network is the closure of the set of achievable rate tuples.
\end{definition}

%We will be particularly interested in coding schemes that have \emph{finite reading precision}.
We will first focus on coding schemes that have \emph{finite reading precision}. 
Then we will show that coding schemes with infinite reading precision can be converted into coding schemes with finite reading precision without much loss in performance.

\begin{definition} \label{precisiondef}
For some $x \in \R$ and a positive integer $\rho$, let $\left\lfloor x \right\rfloor_{\rho} = 2^{-\rho}\lfloor 2^\rho x \rfloor$.
A coding scheme $\C$ is said to have finite reading precision $\rho \in \IN$  %, and is denoted as $\C = (n,\vec R,\rho)$, 
if its relaying functions satisfy
%transmit signal of each (non-source) node $v$ in the network at each time $t$ only depends on
\aln{
%& \left\lfloor Y_v[i] \right\rfloor_\rho \defi 2^{-\rho} \left\lfloor 2^\rho Y_v[i] \right\rfloor, \text{ for $i=1,...,t-1$}, \\
%& f_{s_m,t} ( x_m^n, y^{t-1}) = f_{s_m,t} ( \left\lfloor x_m^n \right\rfloor_{\rho_m}, y^{t-1}),\ \forall\ m\in[1:k] \\
& r_{v}^{(t)} ( y_1,...,y_{t-1}) = r_{v}^{(t)} ( \lfloor y_1\rfloor_\rho,...,\lfloor y_{t-1}\rfloor_\rho),
}
for any $(y_1,...,y_{t-1}) \in \R^{t-1}$, any $v \in V - \{s_1,...,s_L \}$, and any time $t$, and its decoding functions satisfy
\aln{
& g_{i} ( y_1,...,y_{n}) = g_{i} ( \lfloor y_1\rfloor_\rho,...,\lfloor y_{n}\rfloor_\rho),
%f_j \left( \left\lfloor x_1\right\rfloor_\rho,...,\left\lfloor x_n \right\rfloor_\rho \right),
}
for any $(y_1,...,y_{n}) \in \R^{n}$, and $i \in \{1,...,L\}$.
\end{definition}

%
%\begin{definition} \label{finiteprecisiondefn}
%A coding scheme $\mathcal{C}_{n,\rho}$ of blocklength $n$ is said to have
%\textit{finite encoding precision} $\rho=[\rho_1,\cdots,\rho_k]\in\mathbf{N}^k$ if the encoding function at each source $s_m \in \mathcal{S}$ satisfies
%\aln{
%f_{s_m,t} ( x_m^n, y^{t-1}) = f_{s_m,t} ( \left\lfloor x_m^n \right\rfloor_{\rho_m}, y^{t-1}),\ \forall\ m\in[1:k]
%%f_j \left( \left\lfloor x_1\right\rfloor_\rho,...,\left\lfloor x_n \right\rfloor_\rho \right),
%}
%for any $x_m^n \in \R^n$, any $y^{t-1} \in \R^{t-1}$, and any time $t$.
%\end{definition}

%For a sequence of codes with finite reading precision to achieve a rate tuple $\vec R$, we allow the precision of each code on the sequence to vary, and thus we can have the precision tending to infinity as $n \to \infty$.

\begin{definition} \label{precisionachievedef}
Rate tuple $\vec R$ is achievable by coding schemes with finite reading precision if we have a sequence of coding schemes $\C_n$, where coding scheme $\C_n$ has finite reading precision $\rho_n$, which achieves rate tuple $\vec R$ according to Definition \ref{achievedef}.
\end{definition}

\begin{remark}
Notice that we allow the precision $\rho_n$ to vary arbitrarily along the sequence of codes, and it may be the case that $\rho_n \to \infty$ as $n \to \infty$.
\end{remark}

%The assumption of finite reading precision on a coding scheme is certainly mild from a practical point of view: no real communication scheme can require the readings from the antennas to have infinite precision.
%Moreover, as we show in Section \ref{finitesec}, any rate tuple that can be achieved with a sequence of coding schemes with infinite reading precision can also be achieved with a sequence of coding schemes with finite reading precision.

%Moreover, it seems that, even from a purely theoretical point of view, it is not a very restrictive assumption.
%This is because most reasonable schemes can be ``truncated'' to yield a finite reading precision coding scheme with very similar performance, since our finite precision can be chosen arbitrarily large.
%This is discussed in Section \ref{extensionsec}.

\end{section}

\begin{section}{Main Result} \label{mainsec}

Our main result is to show that any rate tuple that is achievable 
%by coding schemes with finite reading precision 
on a network where each $N_v$ is Gaussian for each $v \in V$ is also achievable on a network where each $N_v$ instead has \emph{any} distribution with the same mean and variance. 
In the special case of $L$-unicast wireless networks, our main result is the following theorem.

%Our goal is to show that, by restricting ourselves to codes with finite reading precision, %$N_v \sim \N(0,1)$ 
%the Gaussian noise is the worst-case noise. 
%In other words, we want to show that any rates that are achievable by coding schemes with finite reading precision on a network where each $N_v$ is normally-distributed are achieved on a network where each $N_v$ has an arbitrary continuous distribution with same mean and variance.
%Our main result is the following.

%\begin{theorem}
%Suppose we have a coding scheme that achieves rate vector $\vec R$ on an AWGN network, where the noises are i.i.d. $\N(0,1)$.
%Then it is possible to construct a scheme to achieve arbitrarily close to $\vec R$ on a network where, for each relay $v$, the additive noise at $v$ $N$ is i.i.d. and has any distribution with a continuous cdf satisfying $E[N_v]=0$ and $E\left[N_v^2\right]=1$.
%\end{theorem}

\begin{theorem}[Worst-Case Noise for $L$-Unicast Networks] \label{mainthm} 
From a sequence of coding schemes that achieve rate tuple $\vec R$ on an AWGN $L$-unicast wireless network, it is possible to construct a single sequence of coding schemes that achieves arbitrarily close to $\vec R$ on the same $L$-unicast wireless network, where, for each relay $v$, the distribution of $N_v$ is replaced with any distribution satisfying $E[N_v]=0$ and $E\left[N_v^2\right]=\sigma_v^2$.
Therefore, if $C_{\rm AWGN}$ is the capacity region of the AWGN $L$-unicast wireless network, and $C_{{\rm non\text{-}AWGN}}$ is the capacity region of the same wireless network where, for each relay $v$, the distribution of $N_v$ is replaced with an arbitrary distribution satisfying $E[N_v]=0$ and $E\left[N_v^2\right]=\sigma_v^2$, then
\aln{
C_{\rm AWGN} \subseteq C_{\rm non\text{-}AWGN}.
}
%More specifically, 
%Suppose a rate tuple $\vec R$ is achievable on an AWGN wireless network $(G,L)$.
%Then it is possible to construct a sequence of coding schemes that achieves arbitrarily close to $\vec R$ on the same $|L|$-unicast additive noise 
\end{theorem}

%Theorem \ref{mainthm} is proved in Section \ref{proofl}.
We will prove Theorem \ref{mainthm} using the following two auxiliary results.
%In order to prove our main result, we will first prove the following result concerning codes with finite reading precision.

\begin{theorem} \label{finitelem1}
Suppose a rate tuple $\vec R$ is achievable by coding schemes with finite reading precision on an AWGN $L$-unicast wireless network.
Then it is possible to construct a single sequence of coding schemes that achieves arbitrarily close to $\vec R$ on the same $L$-unicast wireless network where, for each relay $v$, the distribution of $N_v$ is replaced with an arbitrary distribution satisfying $E[N_v]=0$ and $E\left[N_v^2\right]=\sigma_v^2$.
\end{theorem}

\begin{theorem} \label{finitelem2}
Suppose we have a sequence of coding schemes $\C_n$ achieving a rate tuple $\vec R$ on an AWGN network.
Then it is possible to construct a sequence of coding schemes $\C_n^\star$ with finite reading precision that also achieves $\vec R$ on the same AWGN network.
\end{theorem}

It is clear that by combining Theorems \ref{finitelem1} and \ref{finitelem2}, Theorem \ref{mainthm} will follow.
The proof of Theorems \ref{finitelem1} and \ref{finitelem2} will be presented in Section \ref{proofl}.
%Then, in Section \ref{extensionsec}, 
The result in Theorem \ref{mainthm} can be generalized to networks with arbitrary traffic demands. % in Section \ref{extsec}.
By generalizing Definition \ref{codedef} for the case of general traffic demands (which we do in Section \ref{extsec}), we can state our main result as follows.

\begin{theorem}[Worst-Case Noise for Networks with General Traffic Demands] \label{mainthm2}
Suppose a rate tuple $\vec R$ is achievable on an AWGN wireless network with some arbitrary traffic demands $\T$.
Then it is possible to construct a sequence of coding schemes that achieves arbitrarily close to $\vec R$ on the same additive noise wireless network where, for each relay $v$, the distribution of $N_v$ is replaced with an arbitrary distribution satisfying $E[N_v]=0$ and $E\left[N_v^2\right]=\sigma_v^2$.
Therefore, if $C_{\rm AWGN}$ is the capacity region of the AWGN wireless network, and $C_{{\rm non\text{-}AWGN}}$ is the capacity region of the same wireless network where, for each relay $v$, the distribution of $N_v$ is replaced with an arbitrary distribution satisfying $E[N_v]=0$ and $E\left[N_v^2\right]=\sigma_v^2$, then
\aln{
C_{\rm AWGN} \subseteq C_{\rm non\text{-}AWGN}.
}
\end{theorem}

In Section \ref{extsec}, we describe how the proofs of Theorems \ref{finitelem1} and \ref{finitelem2} can be extended to the case of general traffic demands, in order to establish Theorem \ref{mainthm2}.

\end{section}

\begin{section}{Proof of Main Result for $L$-Unicast Wireless Networks} \label{proofl}

In this Section, we will prove Theorems \ref{finitelem1} and \ref{finitelem2}, from which Theorem \ref{mainthm} will follow.
To prove Theorem \ref{finitelem1}, we start by assuming that we have a sequence of coding schemes with finite reading precision designed to achieve a rate tuple $\vec R$ on an AWGN network.
Then, through a series of steps, we will use this sequence of coding schemes to construct another sequence of coding schemes that achieves arbitrarily close to the rate tuple $\vec R$ on the corresponding network where the additive noises are not Gaussian.

%The proof can be divided into three main parts.
A diagram illustrating the proof steps of Theorem \ref{finitelem1} is shown in Fig.~\ref{diagram}.
We start by describing an OFDM-like scheme that is applied to all nodes in the network.
%Notice that, in a wireless networks, all channels are, in general, multiple-access channels, and the inputs to distinct multiple-access channels are coupled because of the broadcast property.
%Nonetheless, our OFDM-like scheme is applied to the transmit and received signals at each node, and effectively converts all channels in the wireless network.
The main idea is that, by applying an Inverse Discrete Fourier Transform (IDFT) to the block of transmit signals of each node, and a Discrete Fourier Transform (DFT) to the block of received signals of each node, we create effective additive noise terms that are weighted averages of the additive noise realizations during that block. %node, we create effective additive noise channels where the noise is a weighted average of the additive noise realizations of all the network uses during that block.
We describe this procedure in detail in Section \ref{ofdmsec}.
Then, in Section \ref{noisesec}, we show that this mixture of noises converges in distribution to a Gaussian additive noise term.
This is done by showing that the weighted average of the noise realizations satisfies Lindeberg's Central Limit Theorem Condition \cite{BillingsleyConvergence}.
Therefore, the OFDM-like scheme effectively produces a network where the noises at each node are dependent across time and approximately Gaussian. %(Approximately Gaussian colored noise Network).
%Next, we need to address the fact that the different effective noise terms are dependent of each other, 
The dependence across time is undesirable since our original coding scheme designed for the AWGN network assumed that the additive noise at each receiver is i.i.d.~over time.
To overcome this problem, in Section \ref{outercodesec}, we apply the OFDM-like scheme over multiple blocks, and then we interleave the effective network uses from distinct blocks.
This effectively creates several blocks in which the network behaves as an Approximately AWGN network (with i.i.d.~noises).
Then our original code for the AWGN network can be applied to each approximately AWGN block.
The fact that this code has finite reading precision guarantees that, when applied to the approximately AWGN block, its error probability is close to its error probability on the AWGN network.
%More formally, this means that, for any choice of messages $\vec w \in \prod_{i=1}^{L} \{1,...,2^{k R_i}\}$, the probability that the joint noise realization $\vec Z$ belongs to the error set $A_{\vec w}$ (i.e., causes an error to occur) is approximately the same in both cases, which follows from the fact that $A_{\vec w}$ can be shown to be a continuity set.
More formally, the error probability of a coding scheme with block length $k$, for a given choice of messages $\vec w \in \prod_{i=1}^{L} \{1,...,2^{k R_i}\}$, can be seen as the probability measure of the error set $A_{\vec w}$ (i.e., the set of noise realizations which causes an error to occur).
As illustrated in Fig. \ref{errorsetfig}, in general, this set could be arbitrarily ill shaped.
However, if the coding scheme has finite reading precision, $A_{\vec w}$ can be shown to be a continuity set, which implies that its measure under similar probability measures cannot change much.
\begin{figure}[ht] 
     \centering
     \subfigure[]{
       \includegraphics[height=28mm]{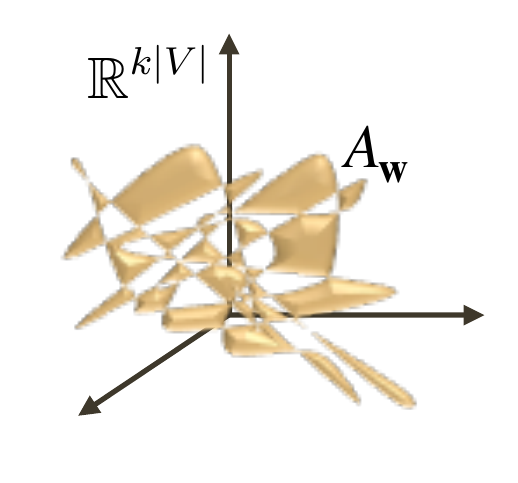}} 
    \hspace{7mm}
    \subfigure[]{
       \includegraphics[height=28mm]{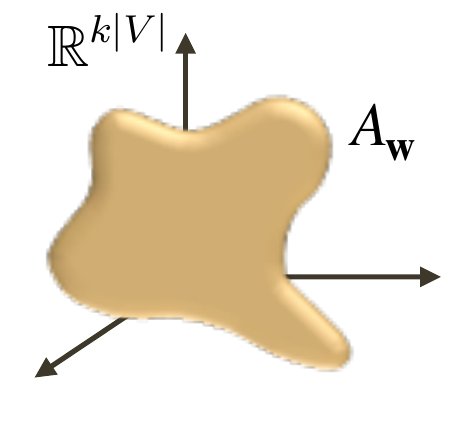}}
	\caption{(a) Illustration of arbitrarily shaped error set $A_{\vec w}$ and (b) continuity set $A_{\vec w}$.\label{errorsetfig}}
\end{figure}
Finally, we take care of the dependence between the noises of different blocks created in the interleaving operation by using a random outer code for each source-destination pair. 
This can be done if we view the coding scheme as creating a discrete channel between the message chosen at a given source and the decoded message at its corresponding destination.
Then we can show via a mutual-information argument that we can use an outer code to achieve a rate tuple arbitrarily close to $\vec R$ on the non-Gaussian wireless network.

In Section \ref{finitesec}, we prove Theorem \ref{finitelem2}.
The main idea is to show that, given a coding scheme with infinite reading precision, there exists a set of quantization mappings, one for each node in the network, such that, if each node quantizes its received signal before applying the relaying or decoding function, the change in the error probability is arbitrarily small.

%\vspace{3mm}

\begin{figure}[ht] 
     \centering
       \includegraphics[height=48mm]{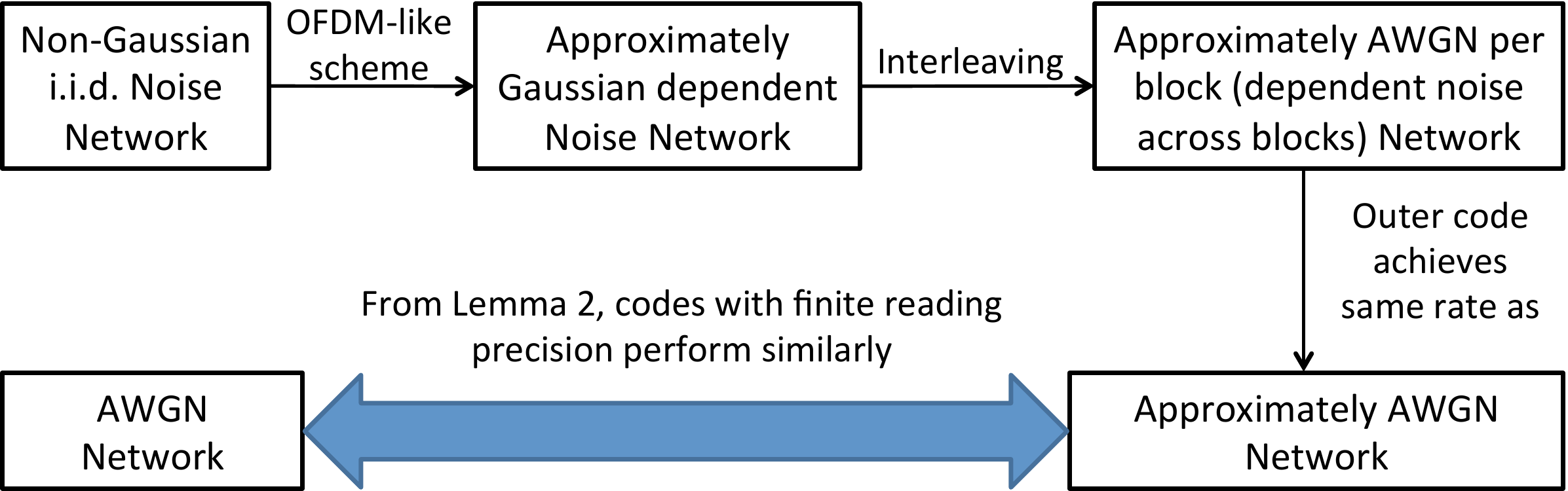} 
	\caption{Diagram of proof steps of Theorem \ref{finitelem1}.
Thin arrows relate to steps in the construction of our new coding scheme, while the thick arrow indicates a conceptual connection established through Lemma \ref{mucontlem}
\label{diagram}}
\end{figure}

We point out that our results are not inconsistent with the intuition that, for a channel with a discrete output alphabet, the worst-case noise should be discrete.
Theorems \ref{finitelem1} and \ref{finitelem2} do {\it not} imply that Gaussian noise is the worst-case noise if we restrict ourselves to coding schemes with finite precision, because, in Theorem \ref{finitelem1}, we may require coding schemes with {\it infinite precision} to achieve the same point in the capacity region in the non-AWGN network (in fact we use coding schemes with infinite precision in our construction based on applying the OFDM-like scheme to the received signals first).

\begin{subsection}{An OFDM-like scheme to mix the noises over time}
\label{ofdmsec}

We use an approach similar to OFDM in order to create an effective network with additive noises that are as close to normally-distributed as we wish.
Essentially, each node in the network will apply transformations to its transmit signals and to its received signals, thus creating an effective network with new input-output relationships.
If we focus on $b$ uses of a single link of the network, then we convert the actual channel (i.e., a mapping from channel inputs $X[0],X[1],...,X[b-1]$ to channel outputs $Y[0],Y[1],...,Y[b-1]$) into an effective channel that maps inputs $d_0,d_1,...,d_{b-1}$ into effective channel outputs $\tilde Y_0, \Re\left[ \tilde Y_1\right], \Im\left[ \tilde Y_1\right], ..., \Re\left[ \tilde Y_{b/2-1}\right], \Im\left[ \tilde Y_{b/2-1}\right], \tilde Y_{b/2}$, where $\Re[z]$ and $\Im[z]$ refer respectively to the real and imaginary parts of a complex number $z$.
The overall transformation, depicted in Fig.~\ref{effective}, can be described as follows.
\begin{figure}[ht] 
     \centering
       \includegraphics[height=45mm]{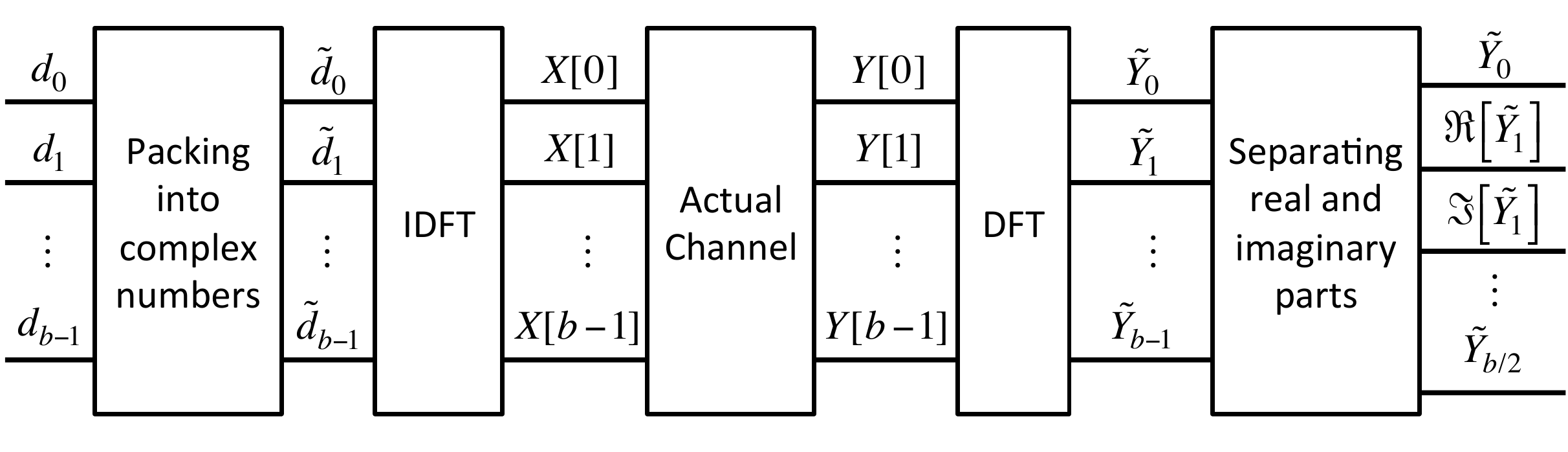} 
	\caption{Diagram of the steps that create the effective channel.\label{effective}}
\end{figure}
Assume that a node $u \in V$ has $b$ real numbers $d_0, d_1,..., d_{b-1}$ which are the inputs to the effective channels we intend to create. %that are meant to be transmitted over an AWGN channel.
We assume that $b$ is even, to simplify the expressions.
Then node $u$ ``packs'' these signals into $b$ complex numbers $\tilde d_0,..., \tilde d_{b-1}$ as follows.
\al{
\begin{array}{ll}
\tilde d_0 = d_0 \\
  \noalign{\medskip}
\tilde d_i = d_{2i-1} + j d_{2i} & \text{ for $i=1,...,\frac b2 - 1$} \\
  \noalign{\medskip}
\tilde d_{b/2} = d_{b-1} \\
  \noalign{\medskip}
\tilde d_i = \tilde d_{b-i}^* & \text{ for $i=\frac b2+1,...,b - 1$} \\
\end{array}
\nonumber
} 
Next, node $u$ takes the IDFT of the vector $\vec{\tilde d_u} = (\tilde d_0 , ..., \tilde d_{b-1})$ to obtain the vector ${\bf X}_u = \IDFT({\bf \tilde d_u})$. 
Throughout the paper, we assume that DFT and IDFT refer to the \emph{unitary} version of the DFT and IDFT.
Since ${\bf \tilde d_u}$ is conjugate symmetric, $\bf X_u$ is a real vector (in $\R^b$). 
%Moreover, if the vector ${\bf d}$ satisfies 
%\al{
%\frac1b \sum_{i=0}^{b-1} E\left[d_i^2\right] \leq P/2, \label{powerdiv2}
%}
Moreover, we will require the original real-valued signals to satisfy
\al{
& \avg \left[d_0^2\right] \leq P, \\
& \avg \left[d_i^2\right] \leq P/2, \text{ for $i=1,...,b-2$},\label{powerdiv2} \\
& \avg \left[d_{b-1}^2\right] \leq P, 
}
where the $\avg$ operator  refers to time average; i.e., if each $d_i$ is seen as a stream of signals $d_i[0],...,d_i[k-1]$, then $\avg(d_i) = k^{-1}\sum_{t=0}^{k-1} d_i[t]$.
%since
%\aln{
%X_n = \frac{1}{\sqrt N}  \sum_{i=0}^{N-1} \tilde d_i \, e^{j2\pi \frac{in}{N}}, \text{ $n=0,...,N-1$},
%}
Then we must have, by Parseval's relationship,
\aln{
%\frac1b \sum_{i=0}^{b-1} E[X_i^2] & = 
\frac1b \, \avg \left[\left\|\vec {X}_u\right\|^2\right] & = 
%\frac1{N^2} E\left[ \sum_{i=0}^{N-1} \sum_{k=0}^{N-1} \sum_{n=0}^{N-1}\tilde d_i \, \tilde d_k^* \, e^{j2 \pi \frac{(i-k)n}{N}}\right] \\
%& = \frac1{N^2} E\left[ \sum_{i=0}^{N-1} \sum_{k=0}^{N-1} \tilde d_i \, \tilde d_k^* \, N \delta_{ik}\right] \\
%& = 
\frac1b \sum_{i=0}^{b-1} \avg \left[\big|\tilde d_i \big|^2\right] \\
& = \frac1b \left\{ \avg \left[d_0^2\right] + \avg \left[d_{b-1}^2\right] + 2 \sum_{i=1}^{b/2-1} \avg \left[d_{2i-1}^2 + d_{2i}^2\right] \right\} 
%= \frac2b \sum_{i=0}^{b-1} E\left[d_i^2\right] 
\leq P. 
}
Therefore, $u$ may transmit $k$ vectors ${\bf X}_u$, each one over $b$ time-slots, and the average power constraint of $P$ over the block $n=kb$ will be satisfied.
The parameter $k$ can be understood as the number of blocks of length $b$ to which we apply the OFDM-like scheme.
A node $v$ %(i.e., either a destination node or a relay node) 
will receive, over each sequence of $b$ time-slots,
\[ {\bf Y}_v =  {\textstyle \sum_{u \in \I(v)} h_{u,v} {\bf X}_u +{\bf N}_v}.\]
By applying a DFT to each block of $b$ received signals, node $v$ will obtain 
\aln{
{\bf \tilde Y}_v = \DFT({\bf Y}_v) =  {\textstyle \sum_{u \in \I(v)} h_{u,v} {\bf \tilde d}_u} + \DFT({\bf N}_v).
}
The transformation induced by the the use of the IDFT on blocks of transmit signals and the DFT on blocks of received signals is illustrated in Fig.~\ref{networks}.

%\newpage

%\vspace{20mm}
%
%{\begin{tikzpicture}[x=0.2\linewidth,y=-0.42\linewidth]
%    \path [use as bounding box] (3.9,-0) rectangle(0.5,0);
%    \path
%(3.79,.21) node {\footnotesize $\vec{\tilde d_1}$}
%(3.58,.24) node {\footnotesize $\vec{\tilde d_2}$}
%(3.57,.37) node {\footnotesize $\vec{\tilde d_3}$}
%(1.72,.21) node {\footnotesize $\vec{X_1}$}
%(1.54,.25) node {\footnotesize $\vec{X_2}$}
%(1.58,.37) node {\footnotesize $\vec{X_3}$}
%(2.35,.43) node {\footnotesize $\vec {Y_4} = \sum_{i=1}^3 h_{i4} \vec {X_i} + \vec {N_4}$}
%(4.43,.43) node {\footnotesize $\vec {\tilde Y_4} = \sum_{i=1}^3 h_{i4} \vec {\tilde d_i} + \DFT(\vec {N_4})$};
%  \end{tikzpicture} }
%  
%  {\begin{tikzpicture}[x=0.2\linewidth,y=-0.42\linewidth]
%    \path [use as bounding box] (3.9,-0) rectangle(0.5,0);
%    \path
%(3.78,-0.67) node {\footnotesize $\vec{\tilde d_1}$}
%(3.65,-.62) node {\footnotesize $\vec{\tilde d_2}$}
%(3.57,-.51) node {\footnotesize $\vec{\tilde d_3}$}
%(1.63,-.67) node {\footnotesize $\vec{X_1}$}
%(1.5,-.62) node {\footnotesize $\vec{X_2}$}
%(1.44,-.51) node {\footnotesize $\vec{X_3}$}
%(2.25,-.43) node {\footnotesize $\vec {Y_4} = \sum_{i=1}^3 h_{i4} \vec {X_i} + \vec {N_4}$}
%(4.4,-.43) node {\footnotesize $\vec {\tilde Y_4} = \sum_{i=1}^3 h_{i4} \vec {\tilde d_i} + \DFT(\vec {N_4})$};
%  \end{tikzpicture} }  
\begin{figure}[ht] 
     \centering
       \includegraphics[height=30mm]{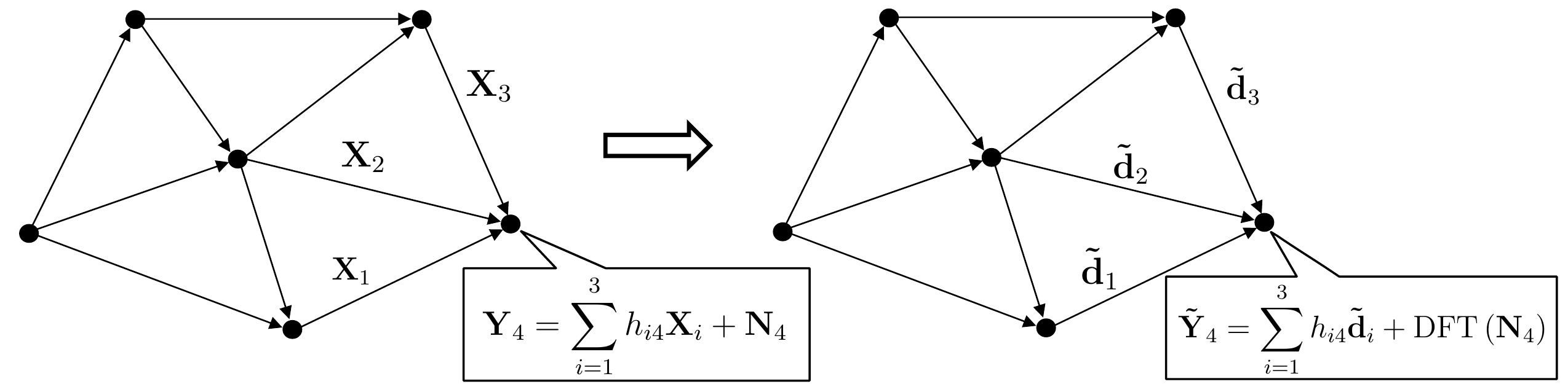} 
	\caption{An illustration of the effect of taking the IDFT of blocks of transmit signals and the DFT of blocks of received signals.\label{networks}}
\end{figure}

Next, by looking at each component of ${\bf \tilde Y}_v$, we notice that we have effectively $b$ complex-valued received signals.
The additive noise on the $\ell$th received signal is given by
%we obtain an effective channel where the additive noise at time-step $k$ is given by 
\al{
\DFT({\bf N}_v)_\ell & = \frac1{\sqrt b} \sum_{i=0}^{b-1} N_v[i] e^{-j2 \pi \frac{i \ell}{b}} \nonumber \\
& = \frac1{\sqrt b} \sum_{i=0}^{b-1} N_v[i] \cos\left( \frac{2 \pi i \ell}{b}\right) - j \, \frac1{\sqrt b} \sum_{i=0}^{b-1} N_{v}[i] \sin\left( \frac{2 \pi i \ell}{b}\right). \label{noiseexp}
}
By considering the real and imaginary parts of each component $\vectil{Y}_{v,i}$ of $\vectil{Y}_{v}$, for $i=0,...,b-1$, separately, we obtain the following $2b-2$ effective real-valued received signals:\al{
\begin{array}{lll}
{\rm (I)} & \vectil{Y}_{v,0} =  \sum_{u \in \I(v)} h_{u,v} {\bf d}_{u,0} + \DFT({\bf N}_v)_0 \\
  \noalign{\medskip}
{\rm (II)} &\Re\left[ \vectil{Y}_{v,i}\right] =  \sum_{u \in \I(v)} h_{u,v} {\bf d}_{u,2i-1} + \Re\left[ \DFT({\bf N}_v)_i \right] & \text{ for $i=1,...,\frac b2 - 1$}  \\
  \noalign{\medskip}
{\rm (III)} & \Im\left[ \vectil{Y}_{v,i}\right] =  \sum_{u \in \I(v)} h_{u,v} {\bf d}_{u,2i} + \Im\left[ \DFT({\bf N}_v)_i \right] & \text{ for $i=1,...,\frac b2 - 1$}  \\
  \noalign{\medskip}
{\rm (IV)} & \vectil{Y}_{v,b/2} =  \sum_{u \in \I(v)} h_{u,v} {\bf d}_{u,b-1} + \DFT({\bf N}_v)_{b/2} \\
  \noalign{\medskip}
{\rm (V)} & \Re\left[ \vectil{Y}_{v,i}\right] =  \sum_{u \in \I(v)} h_{u,v} {\bf d}_{u,2(b-i)-1} + \Re\left[ \DFT({\bf N}_v)_i \right] & \text{ for $i=\frac b2+1,...,b - 1$}  \\
  \noalign{\medskip}
{\rm (VI)} & \Im\left[ \vectil{Y}_{v,i}\right] =  - \sum_{u \in \I(v)} h_{u,v} {\bf d}_{u,2(b-i)} + \Im\left[ \DFT({\bf N}_v)_i \right] & \text{ for $i=\frac b2+1,...,b - 1$} % \\
%  \noalign{\medskip}
%\tilde d_0 = d_0 \\
%  \noalign{\medskip}
%\tilde d_i = d_{2i} + j d_{2i+1} & \text{ for $i=1,...,\frac b2 - 1$} \\
%  \noalign{\medskip}
%\tilde d_{b/2} = d_b \\
%  \noalign{\medskip}
%\tilde d_i = \tilde d_{b-i}^* & \text{ for $i=\frac b2+1,...,b - 1$} \\
\end{array}
\nonumber
} 
However, from the conjugate symmetry of $\DFT({\bf N}_v)$ (since ${\bf N}_v$ is a real-valued vector), we have that $\Re\left[ \DFT({\bf N}_v)_i \right] = \Re\left[ \DFT({\bf N}_v)_{b-i} \right]$ and $\Im\left[ \DFT({\bf N}_v)_i \right] = - \Im\left[ \DFT({\bf N}_v)_{b-i} \right]$, for $i=1,2,...,b-1$, and all the received signals in (V) and (VI) are repetitions (up to a change of sign) of the received signals in (II) and (III).
%but from the conjugate symmetry of ${\bf \tilde d_u}$ and $\DFT({\bf N})$ 
Therefore, we conclude that we have effectively $b$ distinct real-valued received signals with additive noise (i.e., the channels from (I), (II), (III) and (IV), which are the effective channel outputs shown in Fig.~\ref{effective}).
It is important to notice that the additive noise terms are \emph{dependent} across these $b$ received signals.
%These blocks of length $b$ will be the building blocks for the outer code that we construct in Section \ref{outercodesec}.
We also point out that the stricter power constraint in (\ref{powerdiv2}) will not constitute a problem.
The reason is that the effective received signals during the network uses corresponding to (\ref{powerdiv2}), given by (II) and (III), will be shown in the next Section to be subject to a noise with variance $\sigma_v^2/2$ as opposed to $\sigma_v^2$.
Thus, the effective SNR is still $P/\sigma_v^2$.

\end{subsection}

\begin{subsection}{Noise mixture converges to Gaussian Noise}
\label{noisesec}

In this Section, we show that the additive noise terms of the effective received signals we obtained in the previous Section approximate a Gaussian distribution as $b$ gets large.
In the remainder of the paper, we will write $X_n \stackrel{d}{\goesto} X$ to denote that the random variables $X_1,X_2,...$ converge in distribution to $X$, and $X_n \stackrel{p}{\goesto} X$ to denote that the random variables $X_1,X_2,...$ converge in probability to $X$.
We will use the following classical result.

\begin{theorem}[Lindeberg's Central Limit Theorem \cite{billingsley}] \label{lindthm}
Suppose that for each $b = 1,2,...$, the random variables
$
Y_{b,1}, Y_{b,2},..., Y_{b,b}
$
are independent.
In addition, suppose that, for all $b$ and $i \leq b$, $E[Y_{b,i}] = 0$, and let
\al{
s_b^2 = \sum_{i=1}^b E\left[ Y_{b,i}^2 \right].   \label{sbdef}
} 
Then, if for all $\vep > 0$, Lindeberg's condition
\al{
\frac{1}{s_b^2}\sum_{i=1}^b E\left(Y_{b,i}^2 \, \one\left\{|Y_{b,i}|  \geq \vep s_b\right\}\right) \goesto 0 \text{ as $b \goesto \infty$}
\label{lind}
}
holds, we have that
\aln{
\frac{\sum_{i=1}^b Y_{b,i}}{s_b} \stackrel{d}{\goesto} \N(0,1).
}
\end{theorem}

Lindeberg's CLT can be used to prove the following lemma.

\begin{lemma} \label{noiselem}
Let $N[0], N[1], N[2],...$ be i.i.d.~random variables that are zero-mean, have variance $\sigma^2$ %and have an absolutely continuous distribution,
and let 
\al{
Z_b & = \frac1{\sqrt b} \sum_{i=0}^{b-1} N[i] \cos\left( \frac{2 \pi i \ell_b}{b}\right), \label{lemeq}}
for some $\ell_b \in \{1,...,b-1\} \setminus \{b/2\}$.
%Then, each $Z_b$ has a density $f_n$ which converges pointwise almost everywhere to the density of a Gaussian random variable with zero mean and variance $1/2$.
Then, $Z_b$ converges in distribution to $\N(0,\sigma^2/2)$ as $b \to \infty$.
\end{lemma}

\begin{proof}
We start by letting $Y_{b,i+1} = N[i] \cos\left( \frac{2 \pi i \ell_b}{b}\right)$, for $i=0,1,...,b-1$.
Then, by following (\ref{sbdef}), we have
%First we notice that, if $Y_i = N_i \cos\left( \frac{2 \pi i m}{n}\right)$, then we have
\aln{
s_b^2 & = \sum_{i=1}^{b} E\left[Y_{b,i}^2\right] = \sum_{i=0}^{b-1} E\left[N[i]^2\right] \cos^2\left( \frac{2 \pi i \ell_b}{b}\right)  \\
& =  \frac{\sigma^2}{4} \sum_{i=0}^{b-1} \left( e^{j2\pi \ell_b \frac{ i}{b}} + e^{-j2\pi \ell_b \frac{ i}{b}} \right)^2
=  \frac{\sigma^2}{4} \sum_{i=0}^{b-1} \left( e^{j4\pi \ell_b \frac{ i}{b}} + e^{-j4\pi \ell_b \frac{ i}{b}} + 2 \right) \\ 
& =  \frac{b \sigma^2}{2} + \frac{\sigma^2}{4} \sum_{i=0}^{b-1} \left( e^{j4\pi \ell_b \frac{ i}{b}} + e^{-j4\pi \ell_b \frac{ i}{b}} \right) 
 =  \frac{b \sigma^2}{2} + \frac{\sigma^2(1- e^{j4\pi \ell_b})}{4(1-e^{j4\pi \ell_b \frac{1}{b}})} + \frac{\sigma^2(1- e^{-j4\pi \ell_b})}{4(1-e^{-j4\pi \ell_b \frac{1}{b}})} = \frac{b \sigma^2}{2}.
}
The last equality follows because $e^{-j 4 \pi \ell_b} = 1$ and $e^{j 4 \pi \ell_b \frac{1}{b}} \ne 1$ for any $\ell_b \in \{1,...,b-1\} \setminus \{b/2\}$.
Next we let $U_{b,i} = Y_{b,i}^2 \, \one \left\{|Y_{b,i}|  \geq \vep s_b \right\} = Y_{b,i}^2 \, \one \left\{|Y_{b,i}|  \geq \vep \sigma \sqrt\frac{b}{2} \right\}$.
Consider any sequence $i_b$, for $b=1,2,...$, such that $i_b \in \{1,...,b\}$, and any $\delta > 0$.
Then we have that
\aln{
\Pr\left(U_{b,i_b} < \delta\right) & \geq \Pr\left(|Y_{b,i_b}| < \vep \sigma \sqrt{b/2} \right) \\
& \geq \Pr\left(|N[i_b-1]| < \vep \sigma \sqrt{b/2} \right) \\
& = \Pr\left(|N[1]| < \vep \sigma \sqrt{b/2} \right) \goesto 1, \text{ as $b \goesto \infty$},
}
which means that $U_{b,i_b} \stackrel{p}{\goesto} 0$ as $b \goesto \infty$.
Moreover, we have that $|U_{b,i_b}| = U_{b,i_b} \leq N[i_b-1]^2$ for all $b$, and $E\left[N[i_b-1]^2\right] =  \sigma^2 < \infty$.
Next, we notice that $N[i-1] \sim N[1]$ for all $i \geq 1$, which implies that, for any $\tau > 0$,
\aln{
\Pr\left[ |U_{b,i_b}| \geq \tau \right] \leq \Pr\left[ N[i_b-1]^2 \geq \tau \right] = \Pr\left[ N[1]^2 \geq \tau \right].
}
Thus, we can apply the version of the Dominated Convergence Theorem described in pages 338-339 of \cite{billingsley}, to conclude that $E[U_{b,i_b}]\goesto 0$ as $b \to \infty$.
We conclude that
\aln{
\frac{1}{s_b^2}\sum_{i=1}^b E\left(Y_{b,i}^2 \, \one\left\{|Y_i|  \geq \vep s_b\right\}\right) & = \frac{2}{\sigma^2 b} \sum_{i=1}^b E\left[ U_{b,i} \right] \leq \frac{2}{\sigma^2} \max_{1 \leq i \leq b} E\left[ U_{b,i} \right]
 \to 0 \text{ as $b \to \infty$,}
}
and Lindeberg's condition (\ref{lind}) is satisfied for any $\vep > 0$.
Hence, from Theorem \ref{lindthm}, we have that
%\aln{
%\frac{\sum_{i=0}^{b-1} N[i] \cos\left( \frac{2 \pi i \ell}{b}\right)}{\sigma \sqrt{b/2}} \stackrel{d}{\goesto} \N(0,1) \quad
%\Longrightarrow \quad 
%Z_b = \frac{\sigma}{\sqrt2} \frac{\sum_{i=0}^{b-1} N[i] \cos\left( \frac{2 \pi i \ell}{b}\right)}{\sigma \sqrt{b/2}}  \stackrel{d}{\goesto} \N(0,\sigma^2/2).} 
\aln{
\frac{\sum_{i=1}^{b} Y_{b,i}}{\sigma \sqrt{b/2}} \stackrel{d}{\goesto} \N(0,1) \quad
\Longrightarrow \quad 
Z_b = \frac{\sigma}{\sqrt2} \frac{\sum_{i=1}^{b} Y_{b,i}}{\sigma \sqrt{b/2}}  \stackrel{d}{\goesto} \N(0,\sigma^2/2).} 
\end{proof}

Now consider the additive noise term in (II).
It is the real part of (\ref{noiseexp}), which, by Lemma \ref{noiselem}, converges in distribution to $\N(0,\sigma_v^2/2)$, as $b \to \infty$.
Moreover, it is easy to see that Lemma \ref{noiselem} can be restated with sines replacing the cosines, and the same result will hold.
Thus, the additive noise in (III) also converges in distribution to $\N(0,\sigma_v^2/2)$.
%Therefore, the imaginary part of (\ref{noiseexp}) also converges in distribution to $\N(0,\sigma^2/2)$.
%We conclude this Section by noticing that, since in \ref{powerdiv2} we restricted the power used in each of the effective channels to $P/2$, we are obtaining effective channels with $\SNR = P/\sigma^2$.
Finally, for the received signals in (I) and (IV), it is easy to see that the additive noise in (\ref{noiseexp}) only has a real component, and by the usual Central Limit Theorem, it converges in distribution to $\N(0,\sigma_v^2)$.

Notice that, since in (\ref{powerdiv2}) we restricted the power used in the network uses corresponding to (II) and (III) to $P/2$, all of our effective channels have the same $\SNR$ they would have if the transmit signals had power $P$ and the noise variance $\sigma_v^2$.
Therefore, for the network uses corresponding to (II) and (III), we can instead assume that the power constraint is $P$, but all nodes divide their transmit signals by $\sqrt2$ prior to transmission, and multiply their received signals by $\sqrt2$.
This yields the following $b$ effective channels,
\al{
\begin{array}{lll}
{\rm (I)} & \vectil{Y}_{v,0} =  \sum_{u \in \I(v)} h_{u,v} {\bf d}_{u,0} + \DFT({\bf N}_v)_0 \\
  \noalign{\medskip}
{\rm (II')} & \sqrt2 \cdot \Re\left[ \vectil{Y}_{v,i}\right] =  \sum_{u \in \I(v)} h_{u,v} {\bf d}_{u,2i-1} + \sqrt2 \cdot \Re\left[ \DFT({\bf N}_v)_i \right] & \text{ for $i=1,...,\frac b2 - 1$}  \\
  \noalign{\medskip}
{\rm (III')} & \sqrt2 \cdot \Im\left[ \vectil{Y}_{v,i}\right] = \sum_{u \in \I(v)} h_{u,v} {\bf d}_{u,2i} + \sqrt2 \cdot \Im\left[ \DFT({\bf N}_v)_i \right] & \text{ for $i=1,...,\frac b2 - 1$} \\
{\rm (IV)} & \vectil{Y}_{v,b/2} =  \sum_{u \in \I(v)} h_{u,v} {\bf d}_{u,b-1} + \DFT({\bf N}_v)_{b/2} \\
  \noalign{\medskip}
\end{array}
\nonumber
} 
all of which have input power constraint $P$ and additive noise with variance $\sigma_v^2$.
The diagram describing the steps that create the effective channel from Fig.~\ref{effective} can then be updated as shown in Fig.~\ref{effective2}. 
\begin{figure}[ht] 
     \centering
       \includegraphics[height=45mm]{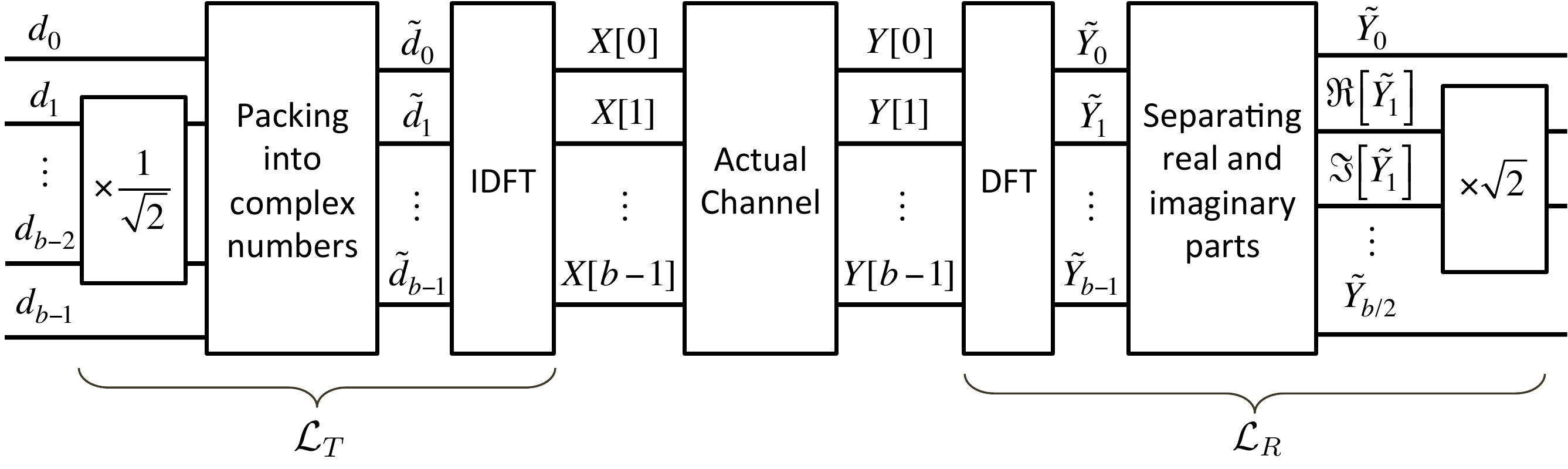} 
	\caption{Diagram of the steps that create the effective channel.
	The overall transformation between the effective channel inputs and the actual channel inputs is represented by a linear transformation $\L_T$ and the overall transformation between the actual channel outputs and the effective channel outputs is represented by a linear transformation $\L_R$. \label{effective2}}
\end{figure}
We notice that the transformation between the $b$ inputs to the effective channels and the $b$ inputs to the actual channel is in fact a 2-norm-preserving linear transformation, which we call $\L_T$.
Similarly, the transformation between the $b$ outputs of the actual channel and the $b$ output of our effective channel is also a 2-norm-preserving linear transformation, which we call $\L_R$.

%In order to simplify the illustration of the scheme we will let $\L_T : \R^b \to \R^b$ and $\L_R : \R^b \to \R^b$ be the overal linear transformations that are applied at blocks of length $b$ transmit signals and received signals in the OFDM-like scheme described in Section~\ref{ofdmsec}. 
%More precisely, $\L_T$ takes $b$ real numbers, packs them into complex numbers and applies an IDFT, obtaining another real-valued vector, while $\L_R$ takes $b$ real numbers, applies a DFT to them, obtaining a conjugate symmetric complex vector, from which $b$ real-valued components are extracted, as described in Section~\ref{ofdmsec}.

Now consider any sequence $\ell_b$, $b=1,2,...$, where $\ell_b \in \{0,...,b-1\}$.
Let $Z_{b,\ell_b}$ now be the additive noise term of the $\ell_b$th effective channel above.
The sequence indices $b \in \{1,2,...\}$ can be partitioned into four sets $J_{1}$, $J_{2}$, $J_{3}$ and $J_{4}$, according to whether $Z_{b,\ell_b}$ corresponds to the additive noise of an effective channel of type (I), (II'), (III') or (IV).
According to Lemma \ref{noiselem}, if $J_2$ or $J_3$ are infinite sets, the subsequence that they define $\{Z_{b,\ell_b}\}_{b \in J_2}$ or $\{Z_{b,\ell_b}\}_{b \in J_3}$ converge in distribution to $\N(0,\sigma_v^2)$ (after the multiplication by $\sqrt2$).
Moreover, as we noticed above, from the usual Central Limit Theorem, it follows that if $J_1$ or $J_4$ are infinite sets, the subsequences defined by $\{Z_{b,\ell_b}\}_{b \in J_1}$ or $\{Z_{b,\ell_b}\}_{b \in J_4}$ also converge in distribution to $\N(0,\sigma_v^2)$. 
Therefore, we conclude that, for any arbitrary sequence $\ell_b$, $b=1,2,...$, where $\ell_b \in \{0,...,b-1\}$, $Z_{b,\ell_b}$ converges in distribution to $\N(0,\sigma_v^2)$.

%Since the effective channels have power constraint of $P/2$, the coding scheme for the AWGN network can be used by dividing all transmit signals by $\sqrt2$.

\end{subsection}

%\newpage

\begin{subsection}{Interleaving and Outer Code}
\label{outercodesec}

%In the previous Section, we saw that by choosing the length of the OFDM block $b$ sufficiently large, it is possible to make the effective additive noise at each node $v$ arbitrarily close (in the distribution sense) to a zero-mean Gaussian noise with variance $\sigma_v^2/2$ for (II) and (III) and $\sigma_v^2$ for (I) and (IV). 
%Notice that, since in (\ref{powerdiv2}) we restricted the power used in the network uses corresponding to (II) and (III) to $P/2$, all of our effective channels have the same $\SNR$ they would have if the transmit signals had power $P$ and the noise variance $\sigma_v^2$. %since the power constraint in each effective channel is $P/2$, we obtain an ${\rm SNR}$ of $P$, for which case the coding scheme was designed.

In this Section, we address the fact that, as we mentioned before, the additive noise at node $v$ in the $b$ effective network uses are dependent of each other.
In order to handle this dependence, we consider using the network for a total of $bk$ times, performing the OFDM-like approach from Section \ref{ofdmsec} within each block of $b$ time steps.
Then, by interleaving the symbols, it is possible to view the result as $b$ blocks of $k$ network uses.
This idea is illustrated in Fig.~\ref{interleavefig}.
\begin{figure}[ht] 
     \centering
       \includegraphics[height=53mm]{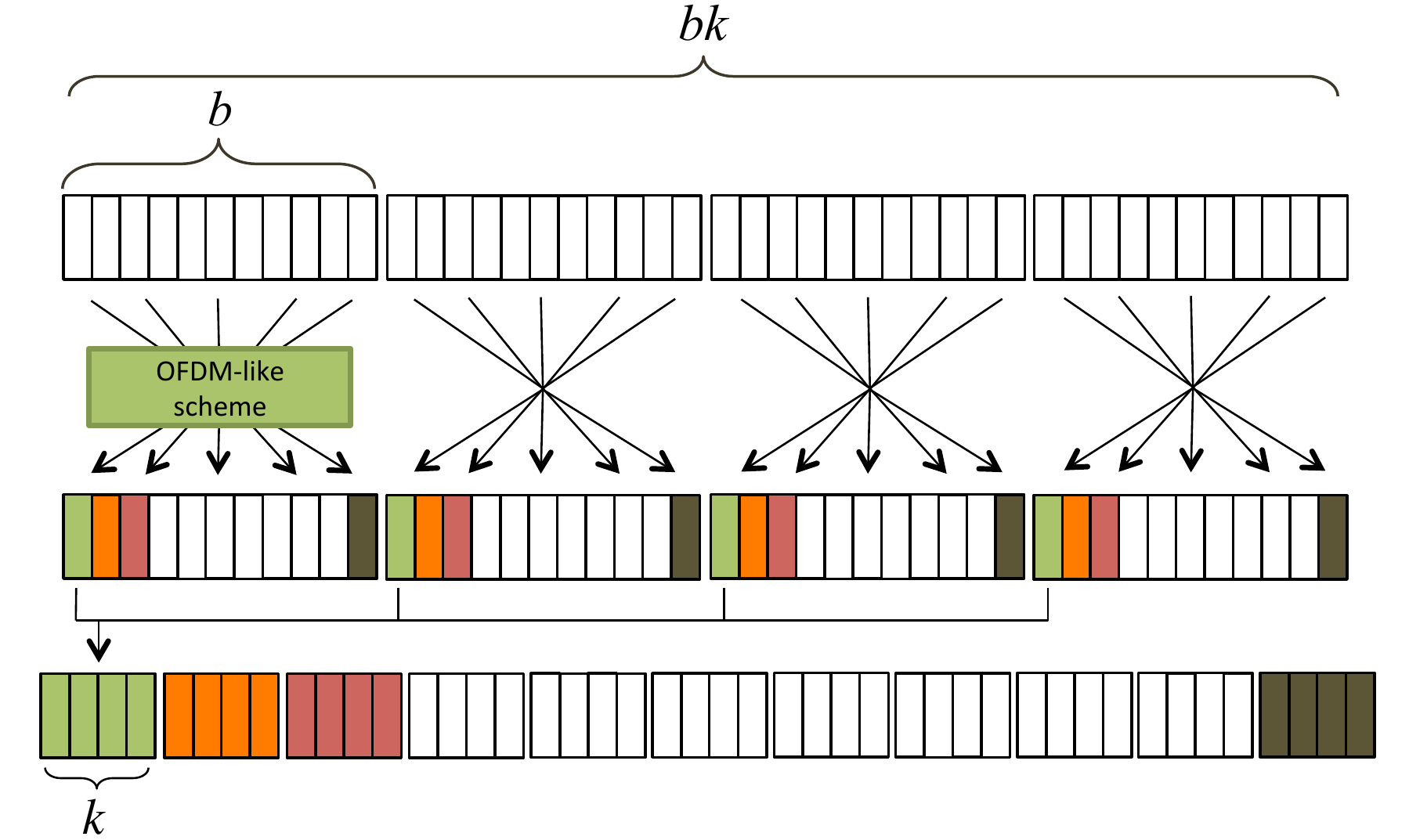} 
	\caption{Interleaving the effective network uses obtained from the OFDM-like scheme.\label{interleavefig}}
\end{figure}
Notice that, within each block of $k$ network uses, the additive noises are i.i.d., but they are dependent among distinct blocks.
Intuitively, this makes each of these blocks of $k$ network uses suitable for the application of a coding scheme $\C_k$ with block length $k$.
The dependence between the noises of different blocks of length $k$ will be handled at the end of this Section, through the application of a random outer code. 
Then, by considering a mutual-information argument, we will show that the performance of the resulting coding scheme on the wireless network with non-Gaussian noises is essentially the same as the performance of the original coding scheme $\C_k$ on the AWGN version of the network.

\begin{example}
Consider a simple relay channel, defined by a graph $G=(V,E)$, where $V = \{s,v,d\}$ and $E=\{(s,v),(s,d),(v,d)\}$. %, as shown in Fig.~\ref{relayfig}.
Suppose we have a coding scheme $\C_k$ of block length $k$ and rate $R$ for this network.
The operations performed by the nodes under this scheme at time $t$ can be illustrated as in Fig.~\ref{relay1}.
\begin{figure}[ht] 
     \centering
       \includegraphics[height=33mm]{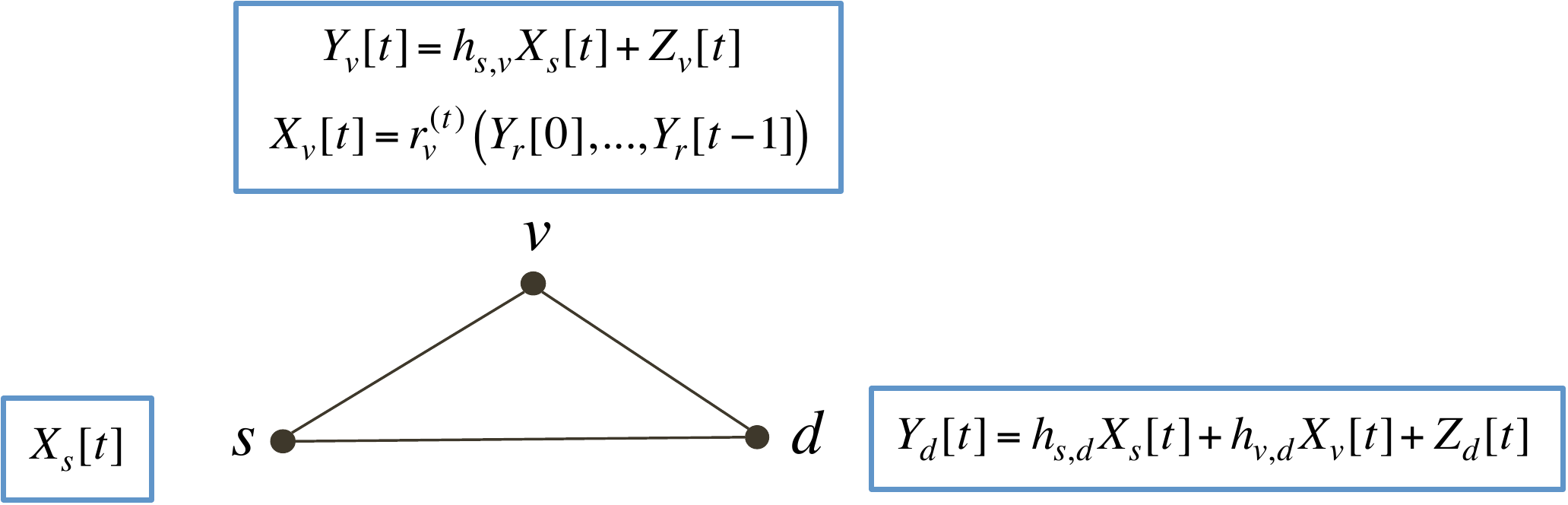} 
	\caption{Illustration of a coding scheme $\C_k$ for a relay channel at time $t$.
At times $t=0,...,k-1$, the source $s$ transmits $X_s[t]$, which is the $(t+1)$th entry of the chosen codeword $f(w)$, for $w \in \{1,...,2^{kR}\}$.
The relay $v$ applies the relaying function $r_v^{(t)}$ to the signals it received up to time $t-1$, to obtain $X_v[t]$, which is then transmitted.
The destination $d$ waits until the end of the length-$k$ block and applies the decoding function $g$ to the block of received signals $(Y_d[0],...,Y_d[k-1])$.
\label{relay1}}
\end{figure}
Now suppose we want to apply the OFDM-like scheme and the interleaving procedure to this coding scheme $\C_k$.
%The resulting coding scheme is illustrated in Fig.~\ref{relay2}.
\begin{figure}[ht] 
     \centering
     \subfigure[]{
       \includegraphics[height=75mm]{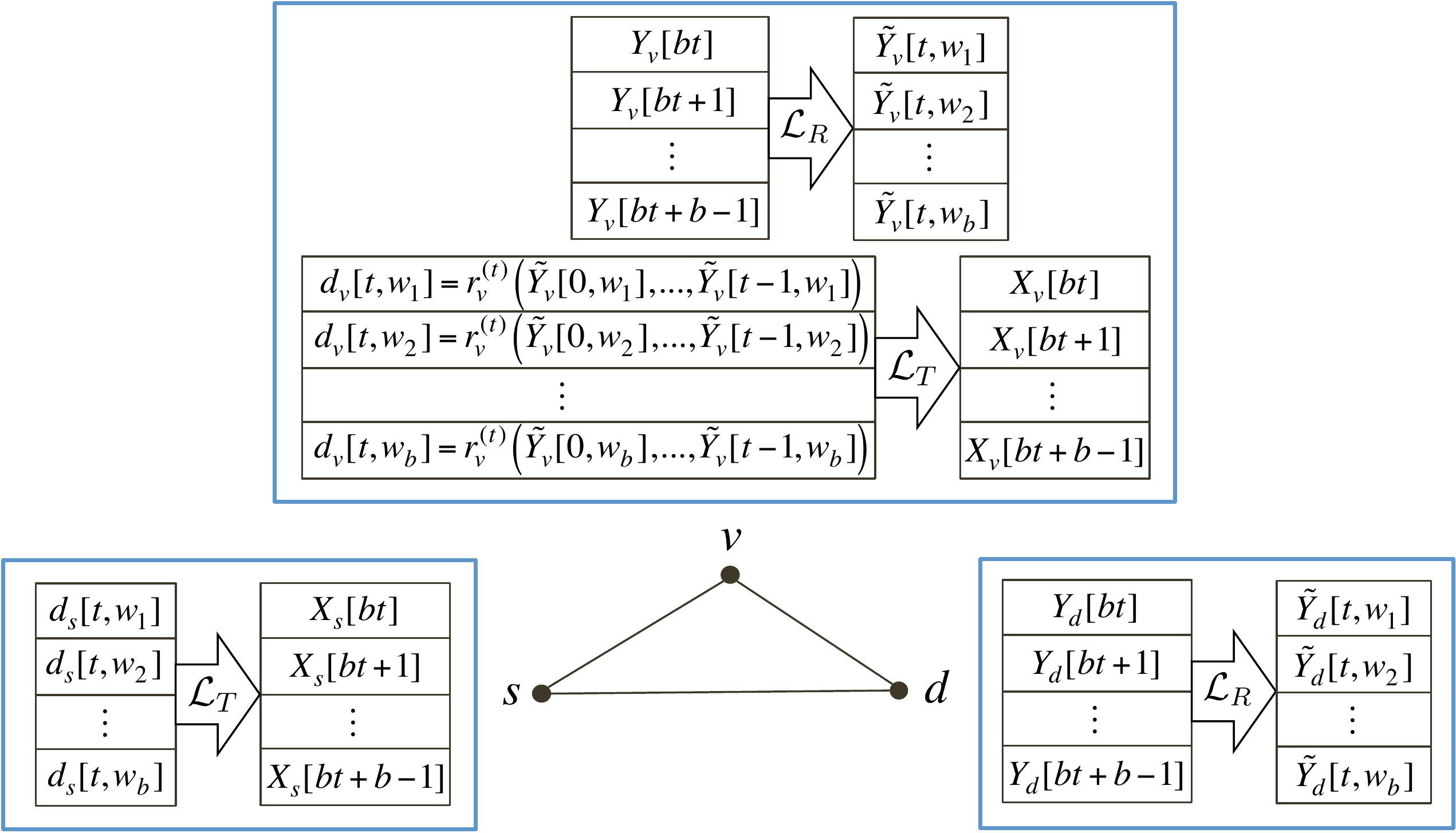} \label{relay2}} 
    \hspace{7mm}
    \subfigure[]{
       \includegraphics[height=35mm]{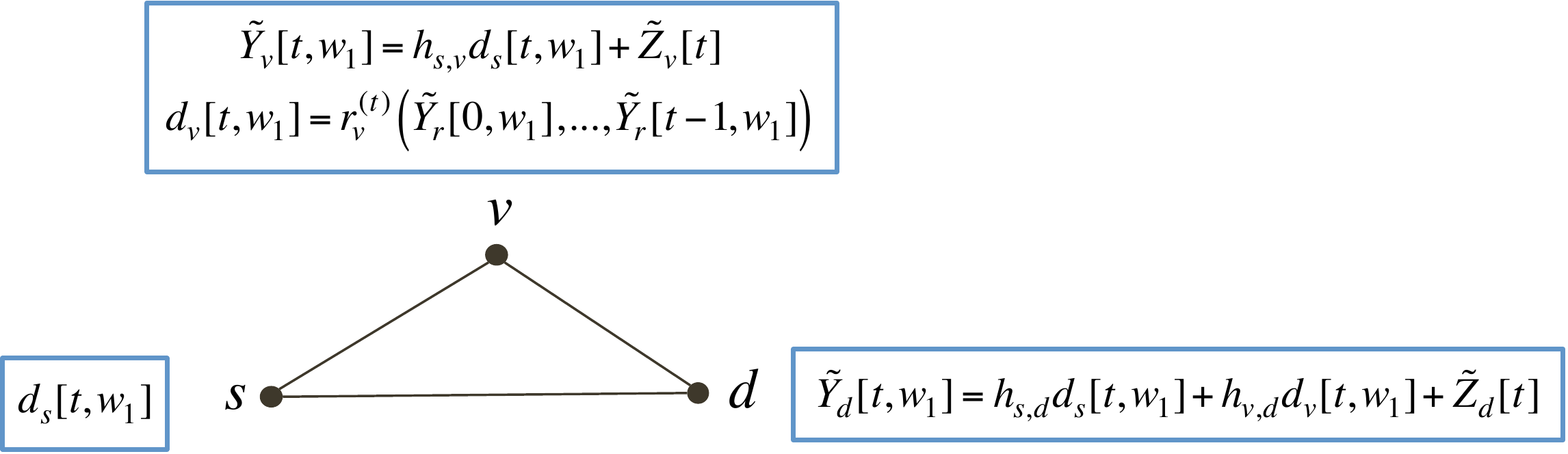} \label{relay3}}
	\caption{(a) Illustration of the source, relay and destination operations, after applying the OFDM-like scheme and the interleaving procedure to a coding scheme $\C_k$. %Operations of a coding scheme $\C_k$ for a relay channel at time $t$.
The source $s$ chooses $b$ messages $w_1,...,w_b \in \{1,...,2^{kR}\}$.
This yields $b$ codewords $f(w_1),...,f(w_b)$ which form the inputs $d_s(t,w_1),...,d_s(t,w_b)$, for $t=0,...,k-1$, to the effective channel.
At times $bt,bt+1,...,bt+b-1$ for $t=0,...,k-1$, $s$ transmits the $b$ signals that result from applying $\L_T$ to the vector $(d_s[t,w_1],...,d_s[t,w_b])$.
At time $bt+b-1$, for $t=0,...,k-1$, the relay $v$ finishes receiving the signals of a length-$b$ block and can apply $\L_R$ to them.
At time $bt$, for $t=1,...,k-1$, using all previously received effective signals, the relay can use relaying function $r_v^{(t)}$ $b$ times to obtain 
$(d_r[t,w_1],...,d_r[t,w_b])$.
After applying $\L_T$ to this vector, the relay obtains the $b$ signals to be transmitted at times $bt,bt+1,...,bt+b-1$.
The destination, at time $bt+b-1$, for $t=0,...,k-1$, finishes receiving the signals of a length-$b$ block and can apply $\L_R$ to them.
(b) Effective network experienced by the signals indexed by $w_1$.
}
\end{figure}
In essence, $b$ versions of this coding scheme will be simultaneously used.
Encoding, relaying and decoding functions are applied ``in parallel'' for each of the $b$ coding schemes, as shown in Fig. \ref{relay2} in detail.
First, $b$ codewords $f(w_1),...,f(w_b)$ are chosen at the source.
At times $bt,bt+1,...,bt+b-1$ for $t=0,...,k-1$, the source transmits the $b$ signals obtained by applying $\L_T$ to the vector formed by the $(t+1)$th entries of these $b$ codewords.
Relay $v$, in turn, after applying $\L_R$ to the received signals at times $bt,bt+1,...,bt+b-1$ for $t=0,...,k-1$, can use the relaying function $r_v^{(t+1)}$ a total of $b$ times in order to obtain a length-$b$ vector that goes through the transformation $\L_T$ to yield the $b$ signals to be transmitted at times $b(t+1),b(t+1)+1,...,b(t+1)+b-1$ for $t=0,...,k-2$.
The destination, after applying $\L_R$ to each block of $b$ received signals, obtains $b$ sequences of $n$ received signals, and can apply its decoding function to each of these sequences.
As shown in Fig. \ref{relay2}, the application of the transformations $\L_T$ and $\L_R$ can be seen as creating $b$ effective networks, where the transmit and received signals of the $i$th effective network are given by $d[t,w_i]$ and $\tilde Y[t,w_i]$ respectively.

%However, the signals that belong to different coding schemes are mixed and demixed through the OFDM-like operations $\L_T$ and $\L_R$, applied to length-$b$ blocks of transmit and received signals respectively.

%Each com

%The sources will now choose $b$ messages from coding scheme $\C_k$.

The purpose of the interleaving procedure can be understood if we focus on what occurs to the signals in one of these effective networks, say the one indexed by $w_1$.
By absorbing the transformations $\L_T$ and $\L_R$ into the network, and viewing the $d[t,w_1]$s and $\tilde Y[t,w_1]$s as inputs and outputs of the network, the network that is effectively experienced by the signals indexed by $w_1$ is shown in Fig.~\ref{relay3}.
%\begin{figure}[ht] 
%     \centering
%       \includegraphics[height=35mm]{relay3.pdf} 
%	\caption{Effective network experienced by the signals indexed by $w_1$.
%\label{relay3}}
%\end{figure}
Notice that the effective network in Fig.~\ref{relay3} is the same as the original network in Fig.~\ref{relay1} but with different additive noise terms $\tilde Z_v[t]$ and $\tilde Z_d[t]$.
These effective noise terms are in fact i.i.d., since the operations $\L_T$ and $\L_R$ are applied to blocks of signals with different indices $w_1,w_2,...,w_b$, and this cannot create dependence between effective noises $\tilde Z_v[t]$ and $\tilde Z_v[t']$ (or $\tilde Z_d[t]$ and $\tilde Z_d[t']$) for $t \ne t'$, since they both correspond to received signals indexed by $w_1$.
Therefore, we are essentially applying coding scheme $\C_k$ in $b$ parallel effective relay channels, each of which has i.i.d.~noises at $v$ and $d$.
\qed

\end{example}

Since from the statement of Theorem \ref{mainthm}, the rate tuple $\vec R$ is achievable by coding schemes with finite reading precision, we may assume that we have a sequence of coding schemes $\C_k$ (with block length $k$ and rate tuple $\vec R$) with finite reading precision $\rho_k$, whose error probability when used on the AWGN network is 
$
\ep_k = P_{\rm error}(\C_k)$,
and satisfies $\ep_k \to 0$ as $k \to \infty$.
Now, consider applying this code over each of the $b$ blocks of length-$k$ that we obtained from the interleaving, as demonstrated in Example 1.
%Notice that, in order to apply code $\C_k$ on a length-$k$ block other than the first or the last one, we will have to divide the output transmit signal of all the nodes by $\sqrt2$ to satisfy (\ref{powerdiv2}), but since the additive noises in these blocks have their variance divided by $2$ as well, each node can re-scale its received signal by multiplying it by $\sqrt 2$, and the code performs in the exact same way.
Over each block of length $k$, the noises at all nodes are independent and i.i.d.~over time, and, if $b$ is chosen fairly large, they are very close to Gaussian in distribution, and, intuitively, the error probability we obtain should be close to $\ep_k$.
The actual distribution of the additive noise at each of these $b$ length-$k$ blocks is given by the noise terms in (I), (II'), (III') and (IV).
For $\ell=0,...,b-1$, we let $\ep_{k,b}^{(\ell)}$ be the error probability of coding scheme $\C_k$ applied on the $(\ell+1)$th such block, for which the i.i.d. additive noise at node $v$ is given by
\aln{
Z_{v,b}^{(\ell)} = \left\{
\begin{array}{lll}
\DFT({\bf N}_v)_0 & \text{ for $\ell=0$}  \\
%  \noalign{\medskip}
\sqrt2 \cdot \Re\left[ \DFT({\bf N}_v)_\ell \right] & \text{ for $\ell=1,...,\frac b2 - 1$}  \\
%  \noalign{\medskip}
\sqrt2 \cdot \Im\left[ \DFT({\bf N}_v)_{(1+\ell-b/2)} \right] & \text{ for $\ell=\frac b2,...,b - 2$} \\
\DFT({\bf N}_v)_{b/2} & \text{ for $\ell=b-1$.}
%  \noalign{\medskip}
\end{array}\right.
\nonumber
} 
Then, for each value of $b$, we let $\ep_{k,b} = \max_{0 \leq \ell \leq b-1} \ep_{k,b}^{(\ell)}$, and $\ell_b = \arg \max_{0 \leq \ell \leq b-1} \ep_{k,b}^{(\ell)}$,
%
%
%More specifically, in each length-$k$ block, the effective noises will be distributed as the noise terms in (I), (II'), (III') and (IV).
%
%
%
%When we apply $\C_k$ to each of these length-$k$ blocks, the error probabilities will be different.
%Thus, for each $b$, we will let $\ep_{b,k}$ be the largest error probability obtained when we apply $\C_k$ to one of the $b$ blocks of length $k$.
%Notice, that if we let $\ell_b \in \{0,...,b-1\}$ be the length-$k$ block for which this maximum error probability is attained, this 
which defines a sequence $\ell_b$, $b=1,2,...$ like the ones considered at the end of Section \ref{noisesec}.

We let $\vec Z_b \in \R^{k|V|}$ be the random vector associated with the effective additive noises at all nodes in $V$ during the $\ell_b$th length-$k$ block assuming that we performed the OFDM-like scheme in blocks of size $b$; i.e.,
\aln{
\vec Z_b = \left( Z_{v,b}^{(\ell_b)}[t]\right)_{v \in V, 0\leq t \leq k-1}.
}
Since each component of $\vec Z_b$ is independent and they all converge in distribution to a zero-mean Gaussian random variable, we have that $\vec Z_b$ converges in distribution to a Gaussian random vector.
We let $\vec Z$ be this limiting distribution, and we know that the component of $\vec Z$ corresponding to node $v$ and time $t$ is distributed as $\N(0,\sigma_v^2)$,
% (or $\N(0,\sigma_v^2/2)$, depending on the length-$k$ block chosen)
 for any $t \in \{0,...,k-1\}$.
Now notice that, if we fix the messages chosen at the sources to be $\vec w = (w_1,w_2,...,w_{L}) \in \prod_{i=1}^{L} \{1,...,2^{k R_i}\}$, then, whether $\C_k$ makes an error is only a deterministic function of $\vec Z_b$.
Therefore, for each ${\bf w} \in \prod_{i=1}^{L} \{1,...,2^{k R_i}\}$, we can define an error set $A_{\bf w}$, corresponding to all realizations of $\vec Z_b$ that cause coding scheme $\C_k$ to make an error.
It is important to notice that $A_{\bf w}$ is independent of the actual joint distribution of the noise terms; it only depends on the coding scheme $\C_k$.
Then we can write
\al{
\ep_{k,b} = 2^{-k \sum_{i=1}^{L} R_i} \sum_{\bf w} \Pr \left[ \vec Z_b \in A_{\bf w} \right] \label{epbk}
}
and also
\al{
\ep_{k} = 2^{-k \sum_{i=1}^{L} R_i} \sum_{\bf w} \Pr \left[ \vec Z \in A_{\bf w} \right]. \label{epbk2}
}
Our first goal is to show that $\ep_{k,b} \to \ep_k$ as $b \to \infty$.
Recall that a Borel set $A \subset \R^m$ is said to be a $\mu$-continuity set for some probability measure $\mu$ on $\R^m$, if $\mu(\partial A) = 0$, where $\partial A$ is the boundary of $A$ (see, for example, \cite{billingsley}).
Next, we state the following classical result, which provides an alternative characterization of convergence in distribution.

\begin{theorem}[Portmanteau Theorem \cite{BillingsleyConvergence}] \label{portmanteau}
Suppose we have a sequence of random vectors $\vec Z_b \in \R^{k|V|}$ and another random vector $\vec Z \in \R^{k|V|}$.
Let $\mu_b$ and $\mu$ be the probability measures on $\R^{k |V|}$ associated to $\vec Z_b$ and $\vec Z$ respectively.
Then $\vec Z_b$ converges in distribution to $\vec Z$ if and only if
\aln{
\lim_{b \to \infty} \mu_b(A) = \mu(A)
}
for all $\mu$-continuity sets $A$.
\end{theorem}

Let $\mu$ be the probability measure on $\R^{k |V|}$ associated to $\vec Z$.
Then, if we show that $A_{\bf w}$ is a $\mu$-continuity set for each choice of messages $\bf w$, from Theorem \ref{portmanteau}, the fact that ${\bf Z}_b \stackrel{d}{\to} \bf Z$ will imply that 
\al{ \label{limAw}
\lim_{b \to \infty} \Pr \left[ \vec Z_b \in A_{\bf w} \right] = \Pr \left[ \vec Z \in A_{\bf w} \right]
}
for each $\bf w$, and from (\ref{epbk}) and (\ref{epbk2}) we will conclude that $\ep_{k,b} \to \ep_k$ as $b \to \infty$.
%which implies that $\ep_{b,k} \to \ep_k$ as $b \to \infty$.
This is in fact what we do in the following Lemma.

\begin{lemma} \label{mucontlem}
Suppose we have a coding scheme $\C$ with block length $k$, rate tuple $\vec R$, and finite reading precision $\rho$.
Then, for any choice of messages ${\bf w} \in \prod_{i=1}^{L} \{1,...,2^{k R_i}\}$, the error set $A_{\bf w}$ is a $\mu$-continuity set.
\end{lemma}

\begin{proof}
Fix some choice of messages $\vec w$.
We will use the fact that $\C$ has finite reading precision $\rho$ to show that our set $A_{\bf w}$ and its complement $A_{\bf w}^c = \R^{k|V|} \setminus A_{\bf w}$ can be represented as a countable union of disjoint convex sets, which will then imply the $\mu$-continuity.
Recall from Definition \ref{precisiondef} that, in a coding scheme with finite reading precision $\rho$, a node $v$ only has access to $\lfloor Y_v \rfloor_\rho$.
Thus, we will call $\lfloor Y_v \rfloor_\rho$ the effective received signal at $v$.
%It is easy to see that the set of all possible values of the effective received signal at $v$ is countable.
%Therefore, if we let 
The set
\aln{
\mathcal{Y} = \left\{ (y_1,...,y_{k|V|}) \in \R^{k|V|} \st y_i = \lfloor y_i \rfloor_\rho, i=1,...,k|V| \right\}
}
can be understood as the set of all possible values of the effective received signals at all nodes in $V$ during a length-$k$ block.
It is clear that $\mathcal{Y}$ is a countable set for any finite $\rho$.
%Next we notice that, since the transmit signal of each relay is a deterministic function of its past effective received signals and the output of the sources is a function of ${\bf w}$, which is fixed, we can define a function $F : \mathcal Y \to \R^{k|V|}$ which defines the transmit signal of each node in $v$ during the length-$k$ block as a function of the effective received signals.
%Then the image of $F$, $\mathcal{X} \defi F(\mathcal Y)$,
%which corresponds to the set of all possible transmit signals at all the nodes during the length-$k$ block, must also be a countable set.

%Notice that, for a fixed choice of messages $\bf w$, what the actual value
%Next, for each $\vec x \in \mathcal X$ and $\vec y \in \mathcal Y$, we define $Q(\vec x, \vec y) \subset \R^{k|V|}$ to be the set of noise values 

Notice that, for our fixed choice of messages $\bf w$, the vector $\vec y \in \mathcal Y$ corresponding to the effective received signals at all nodes during the length-$k$ block is a deterministic function of the value of all the noises in the network during the length-$k$ block, $\vec z \in \R^{k|V|}$.
Therefore, for each $\vec y \in \mathcal Y$, we define $Q(\vec y) \subset \R^{k|V|}$ to be the set of noise realizations $\vec z$ that will result in $\vec y$ being the effective received signals.
In Lemma \ref{qyconvex} in the Appendix, we prove that $Q(\vec y)$ is a convex set.
%We will start with the noise realization $\vec z$, and we will replace one component of $\vec z$ at a time with the corresponding component in $\alpha \vec z + (1-\alpha) \vec z'$.
%After any such replacement, the resulting noise realization will still be in $Q(\vec y)$, and the result will follow.
%Now focus on the received signal at node $v$ at time $\ell$, and let $y_v[\ell]^*$ be the noiseless version of the received signal in its complete binary expansion.
%%entry corresponding to the noise at node $v$ at time $\ell$, we have that
%Then we have that
%\aln{
%\left\lfloor y_v[\ell]^* + z_v[\ell] \right\rfloor_\rho = \left\lfloor y_v[\ell]^* + z'_v[\ell] \right\rfloor_\rho 
%}
%Since any  convex set $S$ is Jordan measurable, it follows that $\lambda(\partial S) = 0$, where $\lambda$ is the Lebesgue measure, for any convex set $S$.
%Next we notice that any convex set is a $\lambda$-continuity set, 
%This follows by noticing that any
We also prove that, for any convex set $S$, $\lambda(\partial S) = 0$, where $\lambda$ is the Lebesgue measure.
Since our measure $\mu$ is absolutely continuous (as $\bf Z$ is jointly Gaussian), it follows by definition \cite{billingsley} that
\aln{
\lambda(S) = 0 \Rightarrow \mu(S) = 0,
}
for any Borel set $S$.
Thus, since $\lambda(\partial Q(\vec y)) = 0$, we have that $\mu(\partial Q(\vec y)) = 0$. %, and $Q(\vec y)$ is a $\mu$-continuity set.
%This, in turn, implies that $\mu\left(Q(\vec y)\right) = \mu\left(Q(\vec y)^\circ\right)$, because
%\al{
%\mu\left(Q(\vec y)\right) \geq \mu\left(Q(\vec y)^\circ\right) & = \mu\left(\overline{Q(\vec y)}\right) - \mu\left(\partial Q(\vec y)\right) \nonumber \\
%& \geq \mu\left({Q(\vec y)}\right) - \mu\left(\partial Q(\vec y)\right) = \mu\left(Q(\vec y)\right). \label{qcircq}
%}
This, in turn, clearly implies that
\al{
\mu\left(Q(\vec y)^\circ\right)  = \mu\left(\overline{Q(\vec y)}\right) = \mu\left(Q(\vec y)\right), \label{qcircq}
}
where we use $S^\circ$ to represent the interior of a set $S$ and $\overline S$ to represent its closure.
Next, let $\mathcal Y_{A_{\bf w}} = \left\{ \vec y \in \mathcal Y \st   A_{\bf w} \cap Q(\vec y) \ne \emptyset \right\}$.
Notice that all noise realizations $\vec z \in Q(\vec y)$ will cause all nodes and, in particular, the destination nodes to receive the exact same effective signals.
Therefore, it must be the case that, if $A_{\bf w} \cap Q(\vec y) \ne \emptyset$, then $Q(\vec y) \subset A_{\bf w}$, which implies that
\aln{
\bigcup_{\vec y \in \mathcal Y_{A_{\bf w}}} Q(\vec y) = A_{\bf w}.
}
Moreover, it is obvious that any noise realization must belong to exactly one set $Q(\vec y)$, and %, if we let $A_{\bf w}^c$ denote the complement of $A_{\bf w}$ in $\R^{k|V|}$, 
we have
\aln{
\bigcup_{\vec y \in \mathcal Y \setminus \mathcal Y_{A_{\bf w}}} Q(\vec y) = A_{\bf w}^c.
}
%Since $Q(\vec y_1) \cap Q(\vec y_2) = \emptyset$ for $\vec y_1 \ne \vec y_2$, it now follows from  (\ref{qcircq}) that
%\aln{
%\bigcup_{\vec y \in \mathcal Y_{A_{\bf w}}} Q(\vec y)^\circ \subset A_{\bf w}^\circ
%}
Finally, we obtain
\aln{ \rescnt
\mu\left(A_{\bf w}^\circ\right) & \geqnum \mu\left(\bigcup_{\vec y \in \mathcal Y_{A_{\bf w}}} Q(\vec y)^\circ \right) \\ 
& \eqnum \sum_{\vec y \in \mathcal Y_{A_{\bf w}}} \mu\left(Q(\vec y)^\circ \right) \eqnum \sum_{\vec y \in \mathcal Y_{A_{\bf w}}} \mu\left({Q(\vec y)} \right) \\
& = 1 - \sum_{\vec y \in \mathcal Y \setminus \mathcal Y_{A_{\bf w}}} \mu\left({Q(\vec y)} \right) = 1 - \sum_{\vec y \in \mathcal Y \setminus \mathcal Y_{A_{\bf w}}} \mu\left({Q(\vec y)^\circ} \right) \\
& = 1- \mu\left(\bigcup_{\vec y \in \mathcal Y \setminus \mathcal Y_{A_{\bf w}}} {Q(\vec y)^\circ} \right) \geq 1 - \mu\left( \left( A_{\bf w}^c\right)^\circ \right) \\ 
& = \mu\left( \left(  \left( A_{\bf w}^c\right)^\circ \right)^c \right) = \mu\left(\overline{A_{\bf w}}\right),
} \rescnt
where \cnt follows since, for sets $B_1,B_2,...$, $\left( \cup_i B_i \right)^\circ \supseteq \cup_i B_i^\circ$,
\cnt follows from the countability of $\mathcal Y_{A_{\bf w}}$ and the fact that $Q(\vec y_1) \cap Q(\vec y_2) = \emptyset$ for $\vec y_1 \ne \vec y_2$, and \cnt follows from (\ref{qcircq}).
We conclude that $\mu(\partial A_{\bf w}) = \mu\left(\overline{A_{\bf w}}\right) - \mu\left(A_{\bf w}^\circ\right)  =0$; i.e., $A_{\bf w}$ is a $\mu$-continuity set.
%
%we notice that, since each noise realization $\vec z \in \R^{k|V|}$ must be in a unique set $Q(\vec y)$, we have that
%
%
%Consider any point $u = \R^{k|V|}$ with the property that each component $u_i$ satisfies
%\aln{
%u_i = 2^{-\rho} \left\lfloor 2^\rho u_i \right\rfloor,
%}
%i.e., $u_i$ has at most $\rho$ non-zero digits in the binary expansion after the decimal point.
%Now, consider the set 
%\aln{
%Q(u) = u + \left\{ x \in \R_+^{k|V|} \st ||x||_{\infty} \leq 2^{-\rho-1} \right\},
%}
%i.e., the set of points obtaining by adding $u$ and a non-negative vector $x$, such that each component of $x$ has no non-zero digits among the first $\rho$ digits after the decimal point in the binary expansion.
%Now, we notice that if $u \in A_{\bf w}$, then $Q(u) \subset A_{\bf w}$, and if $u \notin A_{\bf w}$ then $Q(u) \cap A_{\bf w} = \emptyset$.
%This is true because, since each node
\end{proof}

%Now it follows from Theorem \ref{portmanteau} and Lemma \ref{mucontlem} that, for all message choices ${\bf w}$, we will have
%\al{ \label{limAw}
%\lim_{b \to \infty} \Pr \left[ \vec Z_b \in A_{\bf w} \right] = \Pr \left[ \vec Z \in A_{\bf w} \right],
%}
%which implies that $\ep_{b,k} \to \ep_k$ as $b \to \infty$.

From our previous discussion, we conclude that $\ep_{k,b} \to \ep_k$ as $b \to \infty$.
We then see that we can apply code $\C_k$ within each of the $b$ blocks of length $k$ and obtain a probability of error (within that block) that tends to $\ep_k$ as $b \to \infty$.
However, since we have a total of $b$ blocks of length $k$, we make an error if we make an error in any of the $b$ blocks of length $k$.
It turns out that a simple union bound does not work here, since the error probability would be of the form $b \ep_{k,b}$ and we would not be able to guarantee that it tends to $0$ as $b$ and $k$ go to infinity.
Instead we consider using an outer code for each source-destination pair.

%The next step is to deal with the fact that each of the $b$ blocks of $k$ time-steps have dependent noises.
%
%What we need to do is to use an outer code.

The idea is to apply coding scheme $\C_k$ to each of the $b$ length-$k$ blocks, and then view this as creating a discrete channel for each source-destination pair.
More specifically, for each length-$bk$ block, source $s_j$ chooses a \emph{symbol} (rather than a message) from $\{1,...,2^{k R_j}\}^b$ and transmits the $b$ corresponding codewords from $\C_k$.
Then destination $d_j$ will apply the decoder from code $\C_k$ inside each length-$k$ block and obtain an output symbol also from $\{1,...,2^{k R_j}\}^b$.
%Notice that the probability that the input and output symbols are different, averaged over all input symbols and assuming that the other sources are choosing their input symbols uniformly at random, is exactly $\ep_{b,k}$.
%Each source-destination pair can then be seen as a discrete channel that will be used a total of $b$ times.
Notice that, by viewing the input to $bk$ network uses as a single input to this discrete channel, we make sure we have a discrete \emph{memoryless} channel, and we can use the Channel Coding Theorem.
%Notice that these discrete channels have memory, since from our interleaving scheme, we know that the noise terms from different length $k$ blocks are dependent on each other.
We can view $W_j^b$ and $\hat W_j^b$ as the discrete input and output of the channel between $s_j$ and $d_j$.
We will then construct a code (whose rate is to be determined) for this discrete channel between $s_j$ and $d_j$ by picking each entry uniformly at random from $\{1,...,2^{k R_j}\}^b$.
%and use the fact that each block of $b$ channel uses is independent.
%If we build our random codebooks by picking each entry uniformly at random, 
Then, source-destination pair $(s_j,d_j)$ can achieve rate
\aln{ \rescnt
\frac{1}{bk}I(W_j^b;\hat W_j^b) &= \frac{1}{bk}\left( H(W_j^b) - H(W_j^b | \hat W_j^b) \right)\\
& \geq R_j - \frac{1}{bk} \sum_{\ell=0}^{b-1} H(W_j[\ell] | \hat W_j[\ell]) \\
& \geqnum R_j - \frac1k (1+ \ep_{k,b}^{(\ell)} k R_j) \\
& \geq R_j - \frac1k (1+ \ep_{k,b} k R_j) \\
& = R_j (1- \ep_{k,b}) - \frac1k,
} \rescnt
where \cnt follows from Fano's Inequality, since, within the $\ell$th length-$k$ block, we are applying code $\C_k$ and we have an average error probability of at most $\ep_{k,b}^{(\ell)}$ (it should in fact be less than $\ep_{k,b}^{(\ell)}$ since we are only considering the error event $W_j[\ell] \ne \hat W_j[\ell]$ and $\ep_{k,b}^{(\ell)}$ refers to the union of these events for all source-destination pairs).

We conclude that, by choosing $b$ and $k$ sufficiently large, it is possible for each source-destination pair to achieve arbitrarily close to rate $R_j$. 
Thus, our coding scheme can achieve arbitrarily close to the rate tuple $\vec R$.
This concludes the proof of Theorem \ref{finitelem1}.

\end{subsection}

\begin{subsection}{Optimality of Coding Schemes with Finite Reading Precision} \label{finitesec}

In this Section, we prove Theorem \ref{finitelem2}.
This theorem implies that, if we restrict ourselves to coding schemes with finite reading precision, and allow the reading precision to tend to infinity along the sequence of coding schemes, we can achieve any point in the capacity region of an AWGN wireless network, thus characterizing the optimality of coding schemes with finite reading precision for AWGN networks.
We start by considering a sequence of coding schemes $\C_n$ (with infinite reading precision) that achieves rate tuple $\vec R$ on an AWGN $L$-unicast wireless network.
We will build a sequence of coding schemes $\C_n^\star$ with finite reading precision that also achieves rate tuple $\vec R$ on the same $L$-unicast wireless network.

%Then, we will show that it is possible to have each node in the network quantize its received signals with a sufficiently fine resolution so that the error probability of the resulting coding scheme is close to the error probability of the original infinite precision coding scheme.

Let $\ep_n$ be the error probability of coding scheme $\C_n$, which achieves rate tuple $\vec R$ on the AWGN $L$-unicast wireless network. 
From Definition \ref{achievedef}, we have that $\ep_n \to 0$ as $n \to \infty$.
For any fixed $n$, we will first build a sequence of coding schemes with finite reading precision $\C^\star_{n,m}$, $m=1,2,...$, such that code $\C^\star_{n,m}$ has error probability $\ep_{n,m}$, where $\ep_{n,m} \to \ep_n$ as $m \to \infty$.
This will allow us to choose a finite $m$ for which $\ep_{n,m}$ is arbitrarily close to $\ep_n$.

Notice that, from Definition \ref{codedef}, relaying and decoding functions should be deterministic.
However, in order to construct coding scheme $\C^\star_{n,m}$, we will first assume that the relaying and decoding functions are allowed to be randomized, and later we will derandomize the constructed coding scheme.
Recall that, from Definition \ref{codedef}, coding scheme $\C_n$ is comprised of encoding functions $\left\{f_i : 1 \leq i \leq L \right\}$, relaying functions $\left\{r_v^{(t)} : v \in V, 1 \leq t \leq n \right\}$ and decoding functions $\left\{g_i : 1 \leq i \leq L \right\}$.
We will build $\C^\star_{n,m}$ from $\C_n$ by using the same encoding functions $f_i$, $i=1,...,L$, and replacing the relaying functions with
\aln{
\tilde r_v^{(t)} \left( Y_v[1],..., Y_v[t-1] \right) \defi  r_v^{(t)} \left( \tilde Y^{(m)}_v[1],..., \tilde Y^{(m)}_v[t-1] \right) 
}
for $1 \leq t \leq n$ and $v \in V$, and replacing the decoding functions with
\aln{
\tilde g_i \left( Y_v[1],..., Y_v[n] \right) \defi  g_i \left( \tilde Y^{(m)}_v[1],..., \tilde Y^{(m)}_v[n] \right),
}
for $1 \leq i \leq L$, where we define
%having each relay/destination $v$ apply its relaying/decoding function from coding scheme $\C_n$ to 
\al{ \label{tildey}
\tilde Y^{(m)}_v[t] = \lfloor Y_v[t] \rfloor_m + U_v^{(m)}[t], %\left(-2^{-m-1},2^{-m-1}\right), 
}
for $v \in V$ and $1 \leq t \leq n$, where $U_v^{(m)}[1],...,U_v^{(m)}[n]$ are independent uniform random variables drawn from $\left(-2^{-m-1},2^{-m-1}\right)$, 
%where $U(a,b)$ is a uniform random variable taking values in $(a,b)$, 
independent from all signals and noises in the network.
Notice that, since the relaying functions $r_v^{(t)}$ satisfy the power constraint in Definition \ref{codedef}, so will the new relaying functions $\tilde r_v^{(t)}$.
In order to relate the error probability of $\C_{n,m}^\star$ to the error probability of $\C_n$, we will need the following lemma, whose proof is in the Appendix.

%\begin{lemma}
%Suppose $\vec Y$ is a random vector with density $f$.
%Let $\vec{\tilde Y}_m = \lfloor \vec{Y} \rfloor_m + \vec e_m$, where the quantization operation is taken componentwise, and $\vec e_m$ is a random vector with independent components each distributed as $U\left(-2^{-m-1},2^{-m-1}\right)$.
%Then each $\vec{\tilde{Y}}_m$ has a density $f_m$, and $f_m$ converges pointwise almost everywhere to $f$.
%\end{lemma}

\begin{lemma} \label{convdensity}
Suppose $Y$ is a random variable with density $f$.
Let ${\tilde Y}^{(m)} = \lfloor {Y} \rfloor_m + U^{(m)}$, where $U^{(m)}$ is uniformly distributed in $\left(-2^{-m-1},2^{-m-1}\right)$ and independent from $Y$.
Then each ${\tilde{Y}}^{(m)}$ has a density $f^{(m)}$, and $f^{(m)}$ converges pointwise almost everywhere to $f$.
\end{lemma}

This lemma will be used to show that, by picking $m$ sufficiently large, we can make the error probability of code $\C^{\star}_{n,m}$ arbitrarily close to $\ep_n$.
Suppose we fix the message vector $\vec w  \in \prod_{i=1}^{L} \{1,...,2^{k R_i}\}$ and %we look at joint distribution of the received signals at all the nodes during the $n$ time steps in the block.
let $\vec Y$ be the random vector of length $n |V|$ corresponding to all the received signals at all nodes during the $n$ time steps in the block if code $\C_n$ is used.
More precisely, we write $\vec Y = \left( \vec Y[0], ..., \vec Y[n-1] \right)$, where $\vec Y[t] = (Y_{1}[t],...,Y_{|V|}[t])$ is the random vector of received signals at all $|V|$ nodes at time $t$, for $0 \leq t \leq n-1$.
The received signal at node $v$ at time $t$, $Y_v[t]$, is defined in (\ref{channelmodel}).
Notice that here we assume that the set of nodes $V$ can be written as $V = \{1,...,|V|\}$, in order to simplify some expressions.
We claim that the random vector $\vec Y$ conditioned on the choice of messages $\vec W = \vec w$ has a density.
To see this, we first 
notice that, 
%conditioned on the random message vector being $\vec W = \vec w$, the transmit signals at time $1$, $X_v[1]$ for $v \in V$, are all constant.
%Thus, the received signals $Y_v[1]$, for $v \in V$, are conditionally independent conditioned on $\vec W = \vec w$, and each one is normally-distributed, conditioned on $\vec W = \vec w$.
%Thus, the conditional pdf $f_{Y_v[1]|\vec W}(y_v[1]|\vec w)$ exists for each $v \in V$.
%Similarly, for $t>1$, if we condition 
conditioned on the received signals received up to time $t-1$, i.e., on $(\vec Y[0],...,\vec Y[t-1]) = (\vec y[0],...,\vec y[t-1])$, and on $\vec W = \vec w$, the transmit signals at time $t$, $X_v[t]$ for $v \in V$, are all deterministic.
Thus, the received signals $Y_v[t]$, for $v \in V$, are conditionally independent and each one is normally-distributed, conditioned on $(\vec Y[0],...,\vec Y[t-1]) = (\vec y[0],...,\vec y[t-1])$ and $\vec W = \vec w$.
Therefore, the conditional pdf $f_{Y_v[t]|\vec Y[0],...,\vec Y[t-1],\vec W}(y_v[t]|\vec y[0],..., \vec y[t-1], \vec w)$ exists for each $v \in V$.
We conclude that, conditioned on $\vec W = \vec w$, the random vector $\vec Y$ has a density given by
%\al{ \label{gprod} \rescnt
%f_{\vec Y|\vec W} (\vec y| \vec w) & =  f_{\vec Y_1,...,\vec Y_n} \left(\vec y_{1},...,\vec y_{n}\right)  = f_{\vec Y_1} \left(\vec y_{1}\right) \prod_{t=2}^n f_{\left. \vec Y_t \right| \vec Y_{1},...,\vec Y_{t-1}}\left( \left. \vec y_{t} \right| \vec y_{1},...,\vec y_{t-1} \right) \nonumber \\
%%& \eqnum \prod_{v=1}^{|V|} f_{Y_{1,v}}\left(y_{1,v}\right) \prod_{t=2}^n \prod_{v =1}^{|V|} f_{\left. Y_{t,v} \right| Y_{1,1},...,Y_{1,|V|},...,Y_{t-1,|V|}} \left( \left. y_{t,v} \right| y_{1,1},...,y_{1,|V|},...,y_{t-1,|V|} \right), 
%& \eqnum \prod_{v=1}^{|V|} f_{Y_{1,v}}\left(y_{1,v}\right) \prod_{t=2}^n \prod_{v =1}^{|V|} f_{\left. Y_{t,v} \right| Y_{1,1},...,Y_{1,|V|},...,Y_{t-1,|V|}} \left( \left. y_{t,v} \right| y_{1,1},...,y_{t-1,|V|} \right), 
%} \rescnt
\al{ \label{gprod} \rescnt
f_{\vec Y|\vec W} (\vec y| \vec w) & = \prod_{v=1}^{|V|} f_{Y_{v}[0]|\vec W}\left(y_{v}[0]|\vec w\right) \prod_{t=1}^{n-1} \prod_{v =1}^{|V|} f_{\left. Y_{v}[t] \right| \vec Y[0],..., \vec Y[t-1],\vec W} \left( \left. y_{v}[t] \right| \vec y[0],...,\vec y[t-1],\vec w \right).
} \rescnt
%where \cnt follows since the signals transmitted at time $t$ are a function of the received signals up to time $t-1$, and, thus, conditioned on the received signals up to time $t-1$, the received signals at different nodes at time $t$ are independent (notice that the message vector $\vec w$ is fixed).
%for a given $t$, conditioned on $(\vec Y[1],...,\vec Y[t-1]) = (\vec y[1],...,\vec y[t-1]) \in \R^{t-1}$,  
%Next, we notice that $\vec Y$ can also be written as
%\al{ \label{breaky} 
%\vec Y = \left( \sum_{u \in \I(v)} h_{u,v} X_u[t]\right)_{{1 \leq v \leq |V|, 1\leq t \leq n }} + \Big( Z_u[t] \Big)_{{1 \leq v \leq |V|, 1\leq t \leq n }},
%}
%and, since the $n|V|$ additive noise terms are all independent and each has an absolutely continuous distribution, the noise random vector in (\ref{breaky}) has a density, implying that $\vec Y$ has a density as well \cite{billingsley}, which we denote by $f_{\vec Y}$.
Similarly, we let $\vec {\tilde Y}^{(m)}$ be the vector of $n |V|$ effective received signals (\ref{tildey}) if %message vector $\vec w$ is chosen but 
code $\C^\star_{n,m}$ is used instead, i.e.,  $\vec {\tilde Y}^{(m)} = \left( \vec {\tilde Y}^{(m)}[0], ..., \vec {\tilde Y}^{(m)}[n-1] \right)$, where $\vec {\tilde Y}[t] = \left({\tilde Y}^{(m)}_{1}[t],...,{\tilde Y}^{(m)}_{|V|}[t]\right)$.
By using similar arguments to those that led to (\ref{gprod}), we see that, when we condition on $(\vec {\tilde Y}^{(m)}[0],...,\vec {\tilde Y}^{(m)}[t-1]) = (\vec y[0],...,\vec y[t-1])$, and on $\vec W = \vec w$, the effective received signals ${\tilde Y}^{(m)}_v[t]$, for $v \in V$, are conditionally independent (although not normally-distributed).
Then, using the fact that, 
%Then, when we condition on $(\vec {\tilde Y}^{(m)}[1],...,\vec {\tilde Y}^{(m)}[t-1]) = (\vec y[1],...,\vec y[t-1])$, and on $\vec W = \vec w$, the transmit signals at time $t$, $X_v[t]$ for $v \in V$, are all deterministic, and the effective received signals ${\tilde Y}^{(m)}_v[t]$, for $v \in V$, are conditionally independent (although not normally-distributed).
%To see that the conditional pdf $f_{Y_v[t]|\vec Y[1],...,\vec Y[t-1],\vec W}(y_v[t]|\vec y[1],..., \vec y[t-1], \vec w)$ exists, we notice that
from (\ref{tildey}), ${\tilde Y}^{(m)}_v[t]$ is the sum of two independent random variables %(even when conditioned on $(\vec {\tilde Y}^{(m)}[1],...,\vec {\tilde Y}^{(m)}[t-1]) = (\vec y[1],...,\vec y[t-1])$ and $\vec W = \vec w$), and, since 
and $U_v^{(m)}[t]$ has a density (see page 266 in \cite{billingsley}), we conclude that, conditioned on $\vec W$, $\vec {\tilde Y}^{(m)}[t]$ has a conditional density given by
\al{ \label{gmprod} \rescnt
f_{\vec {\tilde Y}^{(m)}|\vec W} (\vec y| \vec w)  = & \prod_{v=1}^{|V|} f_{{\tilde Y}^{(m)}_{v}[0]|\vec W}\left(y_{v}[0]|\vec w\right) \nonumber \\ 
& \prod_{t=1}^{n-1} \prod_{v =1}^{|V|} f_{{\tilde Y}^{(m)}_{v}[t] | \vec {\tilde Y}^{(m)}[0],..., \vec {\tilde Y}^{(m)}[t-1],\vec W} \left( \left. y_{v}[t] \right| \vec y[0],...,\vec y[t-1],\vec w \right).
} \rescnt
%As we mentioned before, conditioned on $(\vec {\tilde Y}^{(m)}[1],...,\vec {\tilde Y}^{(m)}[t-1]) = (\vec y[1],...,\vec y[t-1])$ and $\vec W = \vec w$ (or just $\vec W = \vec w$, if $t=1$), $Y_v[t]$ is normally-distributed and thus has a density. 
The random variables ${\tilde Y}^{(m)}_v = \lfloor Y_v[t] \rfloor_m + U_v^{(m)}[t]$, for $m=1,2,...$, conditioned on $(\vec {\tilde Y}^{(m)}[0],...,\vec {\tilde Y}^{(m)}[t-1]) = (\vec y[0],...,\vec y[t-1])$ and $\vec W = \vec w$, satisfy the conditions of Lemma \ref{convdensity}, and we have that
\aln{
& f_{{\tilde Y}^{(m)}_{v}[0]|\vec W}\left(y_{v}[0]|\vec w\right) \to f_{Y_{v}[0]|\vec W}\left(y_{v}[0]|\vec w\right) \quad \text{ and } \\
& f_{{\tilde Y}^{(m)}_{v}[t] | \vec {\tilde Y}^{(m)}[0],..., \vec {\tilde Y}^{(m)}[t-1],\vec W} \left( \left. y_{v}[t] \right| \vec y[0],...,\vec y[t-1],\vec w \right) \to \\ 
& \quad \quad \quad \quad \quad \quad f_{\left. Y_{v}[t] \right| \vec Y[0],..., \vec Y[t-1],\vec W} \left( \left. y_{v}[t] \right| \vec y[0],...,\vec y[t-1],\vec w \right),
}
as $m \to \infty$, for $t=2,...,n$ and $v \in V$, for almost all $\vec y \in \R^{n|V|}$.
%
%each (conditional) pdf in (\ref{gmprod}) corresponds to a random variable of the form 
%\al{ \label{cplusz}
%\left\lfloor c + Z \right\rfloor_m + U(2^{-(m+1)}, 2^{-(m+1)}), 
%}
%where $c$ corresponds to the noiseless part of the received signal and is a constant due to the conditioning on previously received signals (or, when $t=1$, due to the fact that $\vec w$ is fixed).
%Moreover, for a (conditional) pdf in (\ref{gmprod}) corresponding to (\ref{cplusz}), there is a (conditional) pdf in (\ref{gprod}) corresponding to a random variable of the form $c + Z$.
%Since $Z$ is normally distributed, $c + Z$ has a density, and by letting $Y=c+Z$, Lemma \ref{convdensity} tells us that each (conditional) pdf in (\ref{gmprod}) converges pointwise almost everywhere to the corresponding (conditional) pdf in (\ref{gprod}).
Therefore, we conclude that $f_{\vec{\tilde Y}^{(m)}|\vec W} \left( \vec y |\vec w \right) \to f_{\vec Y|\vec W} (\vec y|\vec w)$ as $m\to \infty$ for almost all $\vec y \in \R^{n|V|}$ and any $\vec w\in \prod_{i=1}^{L} \{1,...,2^{k R_i}\}$.
%
%It now follows from Lemma \ref{convdensity} that each factor in (\ref{gmprod}) converges almost everywhere to the corresponding factor in (\ref{gprod}).
%Therefore, $g_m$ converges almost everywhere to $g$ as well.

Next we notice that, conditioned on the message vector $\vec W = \vec w$, whether we make an error or not is a function of the received signals at all nodes during the $n$ time steps (it is in fact only a function of the received signals at the destinations).
Thus, there exists a set $E_{\vec w} \subset \R^{n |V|}$ of received signals during the $n$ time steps which cause a decoding error (at any of the decoders).
%Clearly, pointwise a.e.~convergence of the density of $\vec{\tilde Y}^{(m)}$ to the density of $\vec Y$ implies pointwise a.e.~convergence of the  density corresponding to $\vec {\tilde Y}_{\rm dest}^{(m)} \defi \left( {\tilde Y}_{v}^{(m)}[t] : v \text{ is a destination}, 1 \leq t \leq n \right)$, to the density of $\left( Y_v[t] :  v \text{ is a destination}, 1\leq t \leq n \right)$. % (the received signals at the destinations when using code $\C_n$).
We will let $\mu_{\vec w}^{(n)}$ be the probability measure on $\R^{n|V|}$ corresponding to $\vec Y$ (the received signals when using coding scheme $\C_n$) conditioned on $\vec W = \vec w$ and $\mu^{(m,n)}_{\vec w}$ be the probability measure on $\R^{n|V|}$ corresponding to $\vec {\tilde Y}^{(m)}$ (the effective received signals when we use coding scheme $\C^{\star}_{n,m}$) conditioned on $\vec W = \vec w$.
By Scheff\'e's Theorem \cite{billingsley}, we have that
\aln{
\sup_{A \in \B} \left|\mu_{\vec w}^{(n)}(A)-\mu^{(m,n)}_{\vec w}(A)\right| \leq \int_{\R^{n |V|}} \left|f_{\vec Y|\vec W}(\vec y|\vec w) - f_{\vec {\tilde Y}^{(m)}|\vec W}(\vec y|\vec w)\right| d \lambda \to 0, \text{ as $m \to \infty$},
}
where $\B$ is the Borel $\sigma$-field on $\R^{n |V|}$, and $\lambda$ is the Lebesgue measure.
This, in turn, implies that for any choice of messages $\vec w$, we must have $\lim_{m \to \infty} \mu^{(m,n)}_{\vec w} (E_{\vec w}) = \mu^{(n)}_{\vec w}(E_{\vec w})$.
We conclude that
\al{ \label{errorlim}
\ep_{n,m} & = 2^{-n \sum_{i=1}^{L} R_i} \sum_{\bf w} \Pr \left[ \left. \vec {\tilde Y}^{(m)} \in E_{\bf w} \right| \vec W = \vec w \right] \\ 
& = 2^{-n \sum_{i=1}^{L} R_i} \sum_{\bf w} \mu^{(m,n)}_{\vec w} \left(E_{\bf w} \right) \stackrel{m \to \infty}{\longrightarrow} 2^{-n \sum_{i=1}^{L} R_i} \sum_{\bf w} \mu^{(n)}_{\vec w} \left(E_{\bf w} \right) = \ep_n.
%\lim_{m \to \infty} \mu_{n,m}(E_{\vec w}) = \mu_n(E_{\vec w}) = \ep_n.
}
Therefore, we can choose, for each $n$, $m_n$ sufficiently large such that the probability of error of code $\C_{m_n,n}^{\star}$, $\ep_{m_n,n}$, is at most $2 \ep_n$.
Finally, we need to take care of the fact that $\C_{m_n,n}^{\star}$ uses randomized relaying and decoding functions.
First, we notice that if we let $\vec U_m$ be the random vector corresponding to the $n |V|$ samples from $(- 2^{-(m+1)},2^{-(m+1)})$ drawn at the $|V|$ nodes during $n$ time steps, then we can write
\aln{
\ep_{m_n,n} & = 2^{-n \sum_{i=1}^{L} R_i} \sum_{\bf w} \Pr \left[ \left. \vec {\tilde Y}^{(m_n)} \in E_{\bf w} \right| \vec W = \vec w \right] \nonumber \\
& = E \left[ 2^{-n \sum_{i=1}^{L} R_i} \sum_{\bf w} \Pr \left[ \left. \vec {\tilde Y}^{(m_n)} \in E_{\bf w} \right| \vec W = \vec w, \vec U_{m_n} \right] \right].
}
Therefore, there must exist some $\vec u \in \R^{n|V|}$ for which
\aln{
2^{-n \sum_{i=1}^{L} R_i} \sum_{\bf w} \Pr \left[ \left. \vec {\tilde Y}^{(m_n)} \in E_{\bf w} \right| \vec W = \vec w, \vec U_{m_n} = \vec u  \right] \leq \ep_{m_n,n}.
 } 
Thus, we define the coding scheme $\C_n^\star$ by having each node $v$ at time $t$ quantize its received signal with resolution $m_n$, add to it $u_{v}[t]$ (i.e., the entry of $\vec u$ corresponding to node $v$ and time $t$) and then apply the relaying/decoding function from code $\C_n$.
It is then clear that $\C_n^\star$ has deterministic relaying/decoding functions, and its error probability is at most $\ep_{m_n,n} \leq 2 \ep_n$.
%
%    we can employ a simple probabilistic method to state that there must be a set of $n |V|$ mumbers in $(-2^{-m_n-1},2^{-m_n-1})$ such that if each node, instead of adding the random variable $U(-2^{-m_n-1},2^{-m_n-1})$ to its quantized received signal, adds a single fixed number, the probability of error of the resulting (deterministic) coding scheme is still at most $2 \ep_n$.
%This resulting coding scheme will be $\C_n^\star$.
Therefore, the sequence of codes $\C_n^\star$, $n=1,2,...,$ has finite reading precision and achieves the rate tuple $\vec R$. %, which concludes the proof of Lemma \ref{finitelem2}.
%It is then clear that $\C_n^\star$ is a sequence of coding schemes with finite precision achieving the same rate tuple $\vec R$.
\end{subsection}

\end{section}

\begin{section}{Extension to General Traffic Demands} \label{extsec}

One immediate extension of the result in Theorem \ref{mainthm} is to consider wireless networks with general traffic demands.
These could include non-unicast flows such as multicast and broadcast flows.
We again consider an additive noise wireless network described by a directed graph $G=(V,E)$.
%This time, however, each node $v \in V$ is a source and has a message $w(v,D)$ for each set of destinations $D \in \Ps(V)$, where $\Ps(V)$ is the power set of $V$.
%Notice that classical examples of wireless networks such as the Multiple Access Channel and the Broadcast Channel can be cast into this form (by letting the rate of several of these messages to be zero).
This time, we will assume that traffic demands are given by $\T(v,U) = 1$, for all $v \in V$ and $U \subseteq V$.
This way, every node has a message for every subset of the remaining nodes.

By proving the worst-case noise result for a wireless network with such traffic demands, the result is also proved for any other traffic demand $\T'$.
To see this, notice that, if $C_{\T} \subset \R^{V \times \Ps(V)}$ is the capacity region of a wireless network with traffic demand $\T(v,U) = 1$, for all $v \in V$ and $U \subseteq V$, then the capacity region of a wireless network with traffic demand $\T'$ can be written as
%This is because that the capacity region of a wireless network with some traffic demands $\T$ can be written as
\aln{
C_{\T'} = \left\{ {\bf R} \in C_{\T} : R(v,U) = 0 \text{ if } \T'(v,U) = 0 \right\}.
}
Hence, if we prove that 
\aln{
C_{\T,{\rm AWGN}} \subseteq C_{\T,{\rm non\text{-}AWGN}},
}
we also prove that, for any traffic demand $\T'$,
\aln{
C_{\T',{\rm AWGN}} \subseteq C_{\T',{\rm non\text{-}AWGN}}.
}
We can now replace Definition \ref{codedef} with the following.

\begin{definition}
A coding scheme $\C$ with block length $n \in \IN$ and rate tuple $\vec R \in \R^{V \times \Ps(V)}$ for an additive noise wireless network consists of:
\begin{enumerate}[1. ]
%\item An encoding function $f_i : \{1,...,2^{n R_i}\} \to \R^{n}$ for each source $s_i$, $i=1,...,L$, where each codeword $f_i(w_i)$, $w_i \in \{1,...,2^{n R_i}\}$, satisfies an average power constraint of $P$.
\item Encoding/relaying functions $r_v^{(t)} : \R^{t-1} \times \prod_{D \in \Ps(V)} \{1,...,2^{n R(v,D)}\} \to \R$, for $t=0,...,n-1$, for each node $v \in V$, satisfying the average power constraint
%\aln{
%E\left[ \sum_{t=1}^n \left[ r_v^{(t)}(Y_v^{t-1},\vec w_v)\right]^2 \right] \leq n P,
%}
\aln{
\frac 1n \sum_{t=1}^n \left[ r_v^{(t)}(y_1,...,y_{t-1},\vec w_v)\right]^2 \leq P,
}
for all $(y_1,...y_{t-1}) \in \R^{t-1}$ and $\vec w_v \in \prod_{D \in \Ps(V)} \{1,...,2^{n R(v,D)}\}$.
%where $Y_v^{t} = \left( Y_v[1],...,Y_v[t]\right)$, $\vec w_v \in \prod_{D \in \Ps(V)} \{1,...,2^{n R(v,D)}\}$ is the vector of messages to be sent from $v$, and where the expectation is taken assuming that each source chooses all its messages independently and uniformly at random, and taking into account all noises in the network.
\item A decoding function $g_u : \R^n \to \prod_{\substack{v\in V,\\ D \in \Ps(V): u \in D}} \{1,...,2^{n R(v,D)}\}$ for each node $u \in V$.
\end{enumerate}
%We say that $n$ is the block length of the coding scheme $\C$, and $\vec R$ its rate tuple.
\end{definition}

%This message is chosen from $\{1,...,2^{n R(v,D)}\}$, where $R(v,D)$ is the traffic rate from $v$ to $D$.

With this definition of a coding scheme, it is straightforward to extend Definitions \ref{achievedef}, \ref{precisiondef} and \ref{precisionachievedef} to this setting.
We can then generalize Theorem \ref{mainthm} as stated in Theorem \ref{mainthm2}.
%Then we can restate Theorem \ref{mainthm} as follows.

%\begin{thmrep}[Main Result for Wireless Networks with General Traffic Demands] \label{mainthm2}
%Suppose a rate tuple $\vec R$ is achievable on an AWGN wireless network $G$ with some arbitrary traffic demands $\T$.
%Then it is possible to construct a sequence of coding schemes that achieves arbitrarily close to $\vec R$ on the same additive noise wireless network $G$ where, for each relay $v$, the distribution of $N_v$ is replaced with an arbitrary absolutely continuous distribution satisfying $E[N_v]=0$ and $E\left[N_v^2\right]=\sigma_v^2$.
%Therefore, if $C_{\rm AWGN}$ is the capacity region of the AWGN wireless network, and $C_{{\rm non\text{-}AWGN}}$ is the capacity region of the same wireless network where, for each relay $v$, the distribution of $N_v$ is replaced with an absolutely continuous distribution satisfying $E[N_v]=0$ and $E\left[N_v^2\right]=\sigma_v^2$, then
%\aln{
%C_{\rm AWGN} \subseteq C_{\rm non\text{-}AWGN}.
%}
%\end{thmrep}

Theorem \ref{mainthm2} can be proved using essentially the same steps in the proof of Theorem \ref{mainthm}.
From the previous discussion, it suffices to prove this result for traffic demands given by $\T(v,U) = 1$, for all $v \in V$ and $U \subseteq V$.
To re-prove Theorem \ref{finitelem1} in this new setting, we start by applying the OFDM-like scheme to the transmit and received signals of every node exactly as done in Section \ref{ofdmsec}.
Thus, the convergence in distribution of the effective additive noise terms to Gaussian, proved in Section \ref{noisesec}, still holds.
Therefore, we may assume that, as in the beginning of Section \ref{outercodesec}, we have $k$ blocks of $b$ network uses each, and we apply the OFDM-like scheme inside each length-$b$ block.
Next, by interleaving the network uses, we obtain $b$ blocks of length $k$ inside which the network is approximately AWGN.
Furthermore, since we start off with a sequence of coding schemes $\C_k$ with finite reading precision, the proof of Lemma \ref{mucontlem} holds verbatim, except that $\vec w$, the vector of messages chosen, is now a vector in $\prod_{v\in V, D \in \Ps(V)} \{1,...,2^{k R(v,D)}\}$.
Thus, within each length-$k$ block, the probability that any node decodes any of its messages incorrectly (assuming all messages are chosen independently and uniformly at random) is upper bounded by $\ep_{k,b}$, where $\ep_{k,b} \to \ep_k$ as $b\to \infty$ and $\ep_k \to 0$ as $k \to \infty$.

In order to deal with the dependence between the noise realizations of different length-$k$ blocks, we will again consider employing outer codes.
This time, however, instead of having one outer code for each source-destination pair, we will have one outer code for each message $w(v,D)$ (i.e., one outer code for each $v \in V$ and $D \in \Ps(V)$).
Thus, for each $v \in V$ and $D \in \Ps(V)$, we will define a \emph{broadcast} discrete channel with input and output alphabet $\{1,...,2^{k R(v,D)}\}^b$, where $v$ is the source and all nodes in $D$ are the destinations, which are all interested in the same message.
We construct each code by sampling $\{1,...,2^{k R(v,D)}\}^b$ uniformly at random.
Let $W(v,d)^b$ correspond to a random symbol chosen by $v$ uniformly at random from $\{1,...,2^{k R(v,D)}\}^b$, and $\hat W(v,d)_u^b$ be the corresponding output symbol at each node $u \in D$.
For the outer code associated with $v$ and $D$, we can achieve rate
\aln{ \rescnt
\frac{1}{bk}\min_{u \in D} I\left(W(v,D)^b;\hat W(v,D)_u^b\right) &= \frac{1}{bk} \min_{u \in D} \left( H(W(v,D)^b) - H(W(v,D)^b | \hat W(v,D)_u^b) \right)\\
& \geq R(v,D) - \max_{u \in D} \frac{1}{bk} \sum_{\ell=0}^{b-1} H(W(v,D)[\ell] | \hat W(v,D)_u[\ell]) \\
& \geqnum R(v,D) - \max_{u \in D} \frac1k (1+ \ep_{k,b} k R(v,D)) \\
& = R(v,D) (1- \ep_{k,b}) - \frac1k,
} \rescnt
where \cnt follows from Fano's Inequality, since, within each length-$k$ block, we apply code $\C_k$ and we have an average error probability of at most $\ep_{k,b}$. % (it should in fact be much less than $\ep_{b,k}$ since we are only considering the error event $W_j[i] \ne \hat W_j[i]$ and $\ep_{b,k}$ refers to the union of these events for all source-destination pairs).
Therefore, by choosing $b$ and $k$ sufficiently large, our constructed code achieves arbitrarily close to $\vec R$ on the non-Gaussian additive noise wireless network.

The proof of Theorem \ref{finitelem2} holds in this new setting almost verbatim.
The only difference is that we now have one rate for each source $s \in V$ and destination set $D \subset V$ and the message vector $\vec w$ has size $V \times \Ps(V)$; thus, the expressions for the error probability in (\ref{errorlim}) must be modified accordingly.
This concludes the proof of Theorem \ref{mainthm2}.

\end{section}

\begin{section}{Concluding Remarks} \label{conclusion}

In this work, we proved that the Gaussian noise is the worst-case noise in additive noise wireless networks.
This extends the classical result that Gaussian noise is the worst-case noise for point-to-point additive noise channels, which is commonly used as a justification for the modeling of the noise in wireless systems as Gaussian noise.
Thus, we provide formal evidence that this modeling is indeed justified beyond the point-to-point setting.

It is important to highlight the fact that we prove our result by actually constructing a coding scheme that performs well on a non-Gaussian network from a coding scheme designed to perform well on an AWGN network.
This is different from the mutual-information-based proof for point-to-point channels, described in Section \ref{intro}, which relies on the Channel Coding Theorem, and, thus, in random coding arguments. % which only provide existence guarantees.
%Even though we do make use of random coding arguments when we construct the outer codes in Section \ref{outercodesec}, this is done mostly as a way to construct coding schemes whose error probability can be shown to tend to zero.
%As mentioned in Section \ref{outercodesec}, the outer codes must be used since the union bound would only provide an outer bound on the error probability of the form $b \ep_{b,k}$, and we do not have any guarantees on the rate of decay of $\ep_{b,k}$ as $b$ and $k$ increase.
%Thus, it seems reasonable to think that given a sequence of coding schemes for an AWGN network whose error probability decays fast to $0$, no outer code will be necessary.

Another important point about the techniques we introduce is that the only information about the actual noise distributions required for the coding scheme construction are the mean and the variance.
This means that, given a wireless network with \emph{unknown} noise distributions where only the mean and variance can be measured, it is possible to construct a sequence of coding schemes that achieves the capacity of the corresponding AWGN network.

One simple extension of this work is to consider MIMO wireless networks; i.e., wireless networks where each node can have multiple antennas.
It is not difficult to see that the same arguments will hold in this case, and the Gaussian noise can also be seen to be worst-case.
But the tools we developed are in fact also useful for establishing several other worst-case results in different classes of problems.
In particular, the same DFT-based linear transformation followed by an interleaving procedure was used in \cite{wcsourceITW} in order to show that the Gaussian sources are worst-case data sources for distributed compression of correlated sources over rate-constrained, noiseless channels, with a quadratic distortion measure (i.e., in the context of the quadratic $k$-encoder source coding problem).
A similar approach was also taken in \cite{wcjournal}, where the authors consider the problem of communicating a distributed correlated memoryless source over a memoryless network, under quadratic distortion constraints. % with respect to the quadratic distortion measure. 
In this setting they show that, (a) for an arbitrary memoryless network, among all distributed memoryless sources with a particular correlation, Gaussian sources are the worst compressible, that is, they admit the smallest set of achievable distortion tuples, and (b) for any arbitrarily distributed memoryless source to be communicated over a memoryless additive noise network, among all noise processes with a fixed correlation, Gaussian noise admits the smallest achievable set of distortion tuples. 

We observe that establishing the worst-case noise for wireless networks can also be a useful tool in determining the relationship between the capacity regions of the same network under different channel models.
For example, in \cite{deterministickkk}, an additive \emph{uniform} noise network is used as a way to connect the capacity region of Gaussian networks with the capacity region of truncated deterministic networks (first introduced in \cite{ADTFullPaper}).
The worst-case noise result is used first to establish that the capacity region of a Gaussian network with noises distributed as $\N(0,1/12)$ is a subset of the capacity region of the same network with noises distributed uniformly in $\left(-\tfrac12,\tfrac12\right)$.
Then, by noticing that the uniform noise network can be emulated on a network with truncated deterministic channels, it is shown that  the capacity region of the truncated deterministic network where the nodes have slightly more power contains the capacity region of the corresponding Gaussian network.

Finally, we point out that the result in Theorem \ref{finitelem2} is interesting in itself, since it implies that the capacity region when we restrict ourselves to coding schemes with finite reading precision and allow the precision to tend do infinity along the sequence of coding schemes is equal to the unrestricted capacity.
In fact, it is not difficult to change the proof of the theorem in order to prove that $C^{(m)}$, the capacity region when we restrict ourselves to coding schemes where only $m$ bits after the decimal point are available, converges to the unrestricted capacity region $C$, as $m \to \infty$.
Since in any practical wireless system the analog received signals must go through an analog-to-digital converter, this result essentially implies that by increasing the resolution of the analog-to-digital converters used in a wireless network, the capacity region of the practical system is indeed approaching the capacity region of the usual infinite-precision models used in the study of wireless networks.

%
%Therefore, our result suggests that, when designing a coding scheme for a network where the additive noises are non-Gaussian and it may be difficult to prove analytical performance guarantees, the OFDM-like scheme from Section \ref{ofdmsec} together with the interleaving and outer code described in Section \ref{outercodesec} can be used in order to make the network almost AWGN.

\end{section}

\section{Acknowledgements}

The authors would like to thank Professors Gennady Samorodnitsky and Aaron Wagner from Cornell University for helpful discussions.
We also thank the anonymous reviewers for many helpful comments and suggestions.

%\vspace{1.5mm}
%This work is in part supported by NSF grants CAREER-0953117, CCF-1144000, CCF-1161720, AFOSR Young Investigator Program award FA9550-11-1-0064, and NSF TRUST Center.

\appendix 

\section{Appendix} \label{convexapp}

\begin{lemma}
Let $\lambda$ denote the Lebesgue measure.
Then, for any convex set $S$, $\lambda(\partial S) = 0$.
\end{lemma}

\begin{proof}
Consider any point $p \in \partial S$.
Clearly, $p \notin S^\circ$, and by the Supporting Hyperplane Theorem \cite{bertsekasconvex}, there exists a hyperplane that passes through $p$ and contains $S$ in one of its closed half-spaces.
Let $H$ be such a closed half-space.
Since $H$ is closed, it is clear that $\partial S \subset H$.
Then, for any closed ball $B_\ep(p)$ centered at $p$, it is clear that
\aln{
\frac{\lambda(B_\ep(p) \cap \partial S)}{\lambda(B_\ep(p))} %\leq \frac{\lambda(B_\ep(p) \cap \overline S)}{\lambda(B_\ep(p))} 
\leq \frac{\lambda(B_\ep(p) \cap H)}{\lambda(B_\ep(p))} = 1/2.
}
%for $S$ that contains $p$.
%Therefore, at least half of any ball centered at $p$ does not contain any points in $\overline S$.
%We conclude that any point in the boundary of $S$ must have density at most $1/2$, and 
By Lebesgue's Density Theorem, the set 
\aln{
P = \left\{ p \in \partial S \st \liminf_{\ep \to 0} \frac{\lambda(B_\ep(p)  \cap \partial S)}{\lambda(B_\ep(p))} < 1 \right\}
}
should have Lebesgue measure zero.
But since $P = \partial S$, we conclude that $\lambda\left(\partial S\right) = 0$. 
%
%By applying Lebesgue's Density Theorem, we see that $\lambda(X) = 0$.
%Moreover, since for all $p \in \left(\partial S\right) \setminus X$, 
%Therefore, by Lebesgue's Density Theorem, we conclude that $\partial S$ has Lebesgue measure zero.
\end{proof}

\begin{lemma} \label{qyconvex}
In the proof of Lemma \ref{mucontlem}, for each $\vec y \in \mathcal Y$, $Q(\vec y)$ is a convex set.
\end{lemma}

\begin{proof}
Consider two noise realizations $\vec z, \vec z' \in Q(\vec y)$ %, and we show that $\alpha \vec z + (1-\alpha) \vec z'$ is also in $Q(\vec y)$, for any $\alpha \in [0,1]$.
and fix some $\alpha \in [0,1]$.
We will show that if we replace one of the components of $\vec z$ with the corresponding component of $\alpha \vec z + (1-\alpha) \vec z'$, the resulting noise realization $\vec z''$ is still in $Q(\vec y)$.
Then, by using the same argument with $\vec z''$ instead of $\vec z$, another component of $\vec z''$ is replaced with a component $\alpha \vec z + (1-\alpha) \vec z'$, and by repeating this argument, it follows that $\alpha \vec z + (1-\alpha) \vec z'$ is itself in $Q(\vec y)$.
So let us focus on the component corresponding to node $v$ at time $\ell$.
Let $y_v[\ell]^*$ be the noiseless version of the received signal at $v$ at time $\ell$ with its complete binary expansion.
%entry corresponding to the noise at node $v$ at time $\ell$, we have that
Since $\vec z$ and $\vec z'$ result in the same $\vec y$, we have that
\aln{
 y_v[\ell] = \left\lfloor y_v[\ell]^* + z_v[\ell] \right\rfloor_\rho = \left\lfloor y_v[\ell]^* + z'_v[\ell] \right\rfloor_\rho.
}
Now, if we assume wlog that $z_v[\ell] \leq z'_v[\ell]$, we have
\aln{
\left\lfloor y_v[\ell]^* + z_v[\ell] \right\rfloor_\rho \leq \left\lfloor y_v[\ell]^* + \alpha z_v[\ell] + (1-\alpha) z'_v[\ell] \right\rfloor_\rho  \leq  \left\lfloor y_v[\ell]^* + z'_v[\ell] \right\rfloor_\rho.
}
Thus, it follows that $y_v[\ell] = \left\lfloor y_v[\ell]^* + \alpha z_v[\ell] + (1-\alpha) z'_v[\ell] \right\rfloor_\rho$, and by replacing $z_v[\ell]$ with $\alpha z_v[\ell] + (1-\alpha) z'_v[\ell]$, we obtain a noise realization $\vec z''$ that is still in $Q(\vec y)$, and the lemma follows.
\end{proof}

\vspace{3mm}

\begin{lemmarep}{\ref{convdensity}}
Suppose $Y$ is a random variable with density $f$.
Let ${\tilde Y}_m = \lfloor {Y} \rfloor_m + U_m$, where $U_m$ is uniformly distributed in $\left(-2^{-m-1},2^{-m-1}\right)$ and independent from $Y$.
Then each ${\tilde{Y}}_m$ has a density $f_m$, and $f_m$ converges pointwise almost everywhere to $f$.
\end{lemmarep}

\vspace{3mm}

\begin{proof}
Since the density of $U\left(-2^{-m-1},2^{-m-1}\right)$ is $g(x) = 2^m \one{\{x \in \left(-2^{-m-1},2^{-m-1}\right)\}}$, $\tilde Y_m$ will have a density $f_m$ that can be written, for almost all $y$, as
\al{
f_m(y) & = E\left[ g\left(y- \lfloor Y\rfloor_m\right) \right] = 2^m E \left[ \one{\{y - \lfloor Y\rfloor_m \in \left(-2^{-m-1},2^{-m-1}\right)\}} \right] \nonumber \\
& = 2^m \Pr \left[ y - \lfloor Y\rfloor_m \in \left(-2^{-m-1},2^{-m-1}\right) \right] \nonumber \\
& = 2^m \Pr \left[ \lfloor Y\rfloor_m \in \left(y-2^{-m-1},y+2^{-m-1}\right) \right] \nonumber \\
& = 2^m \Pr \left[ \lfloor 2^m Y\rfloor \in  \left(y2^m-1/2,y2^m+1/2\right) \right] \nonumber \\
& = 2^m \Pr \left[ 2^m Y \in  \left( \lceil y2^m-1/2 \rceil, \lceil y2^m+1/2 \rceil \right) \right] \nonumber \\
& = 2^m \Pr \left[ Y \in  \left( 2^{-m} \lceil y2^m-1/2 \rceil, 2^{-m} \lceil y2^m+1/2 \rceil \right) \right] \nonumber \\
& = 2^m \int_{a_m}^{b_m} f(x) dx, \label{integral}
}
where $a_m = 2^{-m} \lceil y2^m-1/2 \rceil$ and $b_m = 2^{-m} \lceil y2^m+1/2 \rceil$.
Notice that we can write $b_m = a_m + 2^{-m}$.
Moreover, we have that
\al{ \label{boundsam}
y-2^{-(m+1)} \leq a_m < y + 2^{-(m+1)},
}
from which we have $a_m \to y$ as $m \to \infty$.
If we let $F(y)$ be the cdf of $Y$, then (\ref{integral}) can be written as
\al{ \label{limeq}
\frac{F(b_m) - F(a_m)}{2^{-m}} = \frac{F(a_m+2^{-m}) - F(a_m)}{2^{-m}} \defi q_m.
}
Our goal is to show that $q_m$ converges to $f(y)$ as $m \to \infty$ for almost all $y$.
Since by assumption $Y$ has an absolutely continuous distribution, $F(y)$ is differentiable almost everywhere, so it suffices to show that $q_m$ converges to $f(y)$ as $m \to \infty$ wherever $F(y)$ is differentiable and the derivative is $f(y)$.
Thus, we focus on a $y$ where $F'(y)=f(y)$.
Suppose by contradiction that $q_m$ does not converge to $f(y)$.
Then there must be an $\ep > 0$ and a subsequence $\{q_{m_i}\}_{i=1}^\infty$, such that one of the following
\al{
& q_{m_i} >  f(y) + \ep \label{up} \\
& q_{m_i} <  f(y) -\ep \label{down}
}
holds for all $i\geq 1$.
Suppose wlog that we have a subsequence $\{q_{m_i}\}_{i=1}^\infty$ for which (\ref{up}) holds for all $i\geq 1$.
We will now pick a further subsequence of $\{q_{m_i}\}_{i=1}^\infty$ in the following way.
First, we choose $K \in \Z_+$ large enough so that $f(y)/K < \ep$, and we define $K$ subsets of $\{1,2,...\}$ as
\aln{
S_j = \left\{ i \geq 1 \st y-2^{-(m_i+1)} + \frac{j-1}{K} 2^{-m_i} \leq a_{m_i} < y-2^{-(m_i+1)} + \frac{j}{K} 2^{-m_i} \right\},
}
for $j=1,2,...,K$.
From (\ref{boundsam}), the sets $S_1,...,S_K$ partition $\{1,2,...\}$, and we must be able to find some $S_j$ that is infinite.
Suppose $|S_t| = \infty$.
Then we have a subsequence $\{q_{m_i}\}_{i \in S_t}$, which we re-index as $\{q_{\ell_i}\}_{i=1}^\infty$.
For each of the elements in this subsequence we have \rescnt
\al{
q_{\ell_i} & = \frac{F(a_{\ell_i}+2^{-\ell_i}) - F(a_{\ell_i})}{2^{-\ell_i}} = \frac{F(a_{\ell_i}+2^{-\ell_i}) - F(y)}{2^{-\ell_i}} + \frac{F(y) - F(a_{\ell_i})}{2^{-\ell_i}} \nonumber \\
& = \frac{a_{\ell_i}+2^{-\ell_i}-y}{2^{-\ell_i}} \frac{F(a_{\ell_i}+2^{-\ell_i}) - F(y)}{a_{\ell_i}+2^{-\ell_i}-y} + \frac{y-a_{\ell_i}}{2^{-\ell_i}} \frac{F(y) - F(a_{\ell_i})}{y-a_{\ell_i}} \nonumber \\
& \leqnum \frac{2^{-\ell_i}(1+t/K-1/2)}{2^{-\ell_i}} \frac{F(a_{\ell_i}+2^{-\ell_i}) - F(y)}{a_{\ell_i}+2^{-\ell_i}-y} + \frac{2^{-\ell_i}(1/2-(t-1)/K)}{2^{-\ell_i}} \frac{F(y) - F(a_{\ell_i})}{y-a_{\ell_i}} \nonumber \\
& = (t/K+1/2) \frac{F(a_{\ell_i}+2^{-\ell_i}) - F(y)}{a_{\ell_i}+2^{-\ell_i}-y} + (1/2-(t-1)/K) \frac{F(y) - F(a_{\ell_i})}{y-a_{\ell_i}}, \label{bound}
} \rescnt
where \cnt follows since $F(y)$ is non-decreasing and ${\ell_i} \in S_t$.
Now, notice that the right-hand side in (\ref{bound}) has a limit, and, by taking the $\limsup$, we obtain
\aln{
\limsup_{i\to \infty} q_{\ell_i} & \leq (t/K+1/2) f(y) + (1/2-(t-1)/K) f(y) \\
& = \left(1+\frac1K\right) f(y) < f(y) + \ep.
}
But this is a contradiction because all $q_{m_i}$ satisfied $q_{m_i} > f(y) + \ep$, and $\{q_{\ell_i}\}_{i=1}^\infty \subset \{q_{m_i}\}_{i=1}^{\infty}$.
We conclude that we must have
\aln{
\lim_{m\to \infty} q_m = f(y),
}
which implies that $f_m(y) \to f(y)$ as $m\to \infty$.
%
%
%Similar to (\ref{bound}), one can obtain
%\al{
%q_{\ell_i} & \geq \frac{2^{-\ell_i}(1+t/K-1/2)}{2^{-\ell_i}} \frac{F(a_{\ell_i}+2^{-\ell_i}) - F(y)}{a_{\ell_i}+2^{-\ell_i}-y} + \frac{2^{-\ell_i}(1/2-(t-1)/K)}{2^{-\ell_i}} \frac{F(y) - F(a_{\ell_i})}{y-a_{\ell_i}} \nonumber \\
%& \geq (1/2+(t-1)/K) \frac{F(a_{\ell_i}+2^{-\ell_i}) - F(y)}{a_{\ell_i}+2^{-\ell_i}-y} + (1/2-t/K) \frac{F(y) - F(a_{\ell_i})}{y-a_{\ell_i}}. \label{bound2}
%}
%By taking the $\liminf$, we obtain
%\al{
%\liminf_{i\to \infty} q_{m_i} & \geq (1/2+(t-1)/K) f(y) + (1/2-(t-1)/K) f(y) \\
%& = \left(1+\frac1K\right) f(y).
%}
%
%\vspace{25mm}
%
%Since we have that 
%\aln{
%f(y) = \lim_{\substack{a \to 0\\b\to 0}} \frac{\int_{-\infty}^{y+b} f(x) dx - \int_{-\infty}^{y-a} f(x) dx}{a+b} = \lim_{\substack{a \to 0\\b\to 0}} \frac{\int_{y-a}^{y+b} f(x) dx}{a+b}
%}
%for almost all $y$,
%$a_m \to 0$, $b_m \to 0$ and $ a_m + b_m = 2^{-m} \left( \lceil y2^m-1/2 \rceil + \lceil y2^m+1/2 \rceil \right) = 2^{-m}$, we conclude that
%$f_m(y) \to f(y)$, as $m \to \infty$, for all $y$.
%$h_m = b_m-y = 2^{-m} \lceil y2^m+1/2 \rceil - y$, and $h_m \to 0$ as $m \to \infty$.
%Moreover, $y+h_m = b_m$ and 
%\aln{
%y-h_m = y - b_m + y = 2y - 2^{-m} \lceil y2^m+1/2 \rceil
%}
\end{proof}

\bibliographystyle{unsrt}

\bibliography{refs}

\end{document}